\documentclass[10pt,journal]{IEEEtran}

\usepackage{amssymb}
\usepackage{amsmath}
\usepackage{cite}
\usepackage{url}
\usepackage{xcolor}
\usepackage{cite,graphicx,amsmath,amssymb}
\usepackage{subfigure}
\usepackage{fancyhdr}
\usepackage{mdwmath}
\usepackage{mdwtab}
\usepackage{caption}
\usepackage{amsthm}
\usepackage{setspace}
\usepackage{hyperref}
\usepackage{algorithm}
\usepackage{algorithmic}
\usepackage{pdfpages}
\usepackage{mathtools}
\hypersetup{colorlinks=false}

\graphicspath{{./src}}

\newtheorem{remark}{Remark}
\newtheorem{theorem}{Theorem}

\newtheorem{lemma}{Lemma}

\newtheorem{corollary}{Corollary}

\newtheorem{proposition}{Proposition}

\allowdisplaybreaks

\title{NOMA-aided Joint Communication, Sensing, and Multi-tier Computing Systems}
\author{

        Zhaolin~Wang,~\IEEEmembership{Graduate Student Member,~IEEE,}
        Xidong~Mu,~\IEEEmembership{Member,~IEEE,}\\
        Yuanwei~Liu,~\IEEEmembership{Senior Member,~IEEE,}
        Xiaodong~Xu,~\IEEEmembership{Senior Member,~IEEE,}\\
        and Ping Zhang,~\IEEEmembership{Fellow,~IEEE}

\thanks{Part of this paper will be presented at the IEEE Global Communications Conference, Rio de Janeiro, Brazil, Dec. 4-8, 2022}
\thanks{Zhaolin Wang, Xidong Mu, and Yuanwei Liu are with the School of Electronic Engineering and Computer Science, Queen Mary University of London, London E1 4NS, U.K. (e-mail: zhaolin.wang@qmul.ac.uk, xidong.mu@qmul.ac.uk, yuanwei.liu@qmul.ac.uk).}
\thanks{Xiaodong Xu and Ping Zhang are with State Key Laboratory of Networking and Switching Technology, Beijing University of Posts and Telecommunications, Beijing 100876, China. (e-mail: xuxiaodong@bupt.edu.cn, pzhang@bupt.edu.cn).}
}

\begin{document}

\maketitle
\vspace{-1cm}

\begin{abstract}
    A non-orthogonal multiple access (NOMA)-aided joint communication, sensing, and multi-tier computing (JCSMC) framework is proposed. In this framework, a multi-functional base station (BS) simultaneously carries out target sensing and provide edge computing services to the nearby users. To enhance the computation efficiency, the multi-tier computing structure is exploited, where the BS can further offload the computation tasks to a powerful Cloud server (CS). The potential benefits of employing NOMA in the proposed JCSMC framework are investigated, which can maximize the computation offloading capacity and suppress inter-functionality interference. Based on the proposed framework, the transmit beamformer of the BS and computing resource allocation among the BS and CS are jointly optimized to maximize the computation rate subject to the communication-computation causality and the sensing quality constraints. Both partial and binary computation offloading modes are considered: 1) For the partial offloading mode, a weighted minimum mean square error based alternating optimization algorithm is proposed to solve the corresponding non-convex optimization problem. It is proved that a Karush–Kuhn–Tucker optimal solution can be obtained; 2) For the binary offloading mode, the resultant highly-coupled mixed-integer optimization problem is first transformed to an equivalent but more tractable form. Then, the reformulated problem is solved by utilizing the alternating direction method of multipliers approach to obtain a nearly optimal solution. Finally, numerical results verify the effectiveness of the proposed algorithms and reveal that: i) the computation rate can be significantly enhanced by exploiting the multi-tier computing architecture when the BS is resource-limited, and ii) the proposed NOMA-aided JSCMC framework is superior in inter-functionality interference management and can achieve high-quality sensing and computing performance simultaneously compared with other benchmark schemes.
\end{abstract}


\section{Introduction}

Recently, with the rapid evolution of emerging services like autonomous driving, smart cities, industrial Internet of things (IIoT), and Metaverse, the network nodes in 6G are envisioned to go beyond communication and should carry out multiple functionalities in an integrated manner, such as the high-performance low-latency computing and high-precision environmental sensing \cite{latva2020key}. In the past decade, mobile edge computing (MEC), has long been regarded as an effective solution to provide Cloud-like computing services to mobile devices \cite{mao2017survey}. However, to address the challenges of growing computing demand in future wireless networks, multi-tier computing (MTC) is required to leverage the computing resources along the continuum from Cloud to mobile devices, which enables flexible and powerful computing services through collaboration between different tiers \cite{yang2019multi}. Meanwhile, wireless radio sensing, where the characteristic of the environment is extracted from the radio frequency signals, is regarded as a promising sensing solution for future networks, due to its striking similarities with wireless communications in terms of hardware facilities and signal processing pipelines. Consequently, the recent research theme of integrated sensing and communication (ISAC) has received heated discussion, with the aim of pursuing high resource efficiency and the mutual benefits between the two functionalities \cite{liu2021integrated}. In the above context, future network nodes will go beyond merely providing communication service and evolve towards the integration of communication, sensing, and computing \cite{feng2021joint, zhangtowards}, which brings new challenges of complicated resource coordination and interference management between the three functions. 

Non-orthogonal multiple access (NOMA) has been recognized as a promising multiple access technique for next generation wireless networks \cite{liu2022evolution}. By exploiting the superposition coding (SC) and successive interference cancellation (SIC) at the transmitter and the receiver, respectively, NOMA allows multiple users to share the same resource block. As such, the favorable characteristics of NOMA in terms of resource allocation and interference mitigation have been extensively investigated in various scenarios \cite{liu2016cooperative, zhu2017non, mu2020exploiting, chen2021irs, wang2022iterative}, which is appealing to deal with the aforementioned challenges of integrating the communication, sensing, and computing functionalities. Hence, it is highly valuable to evaluate the potential benefits of NOMA in multi-functional wireless networks.


\subsection{Related Works}
\subsubsection{Studies on MEC and MTC}
Due to the favorable characteristics of MEC, it has been widely discussed upon energy efficiency maximization \cite{you2016energy}, end-to-end latency minimization \cite{mao2016dynamic}, and computation rate maximization \cite{zhou2018computation}. Nevertheless, the aforementioned works put their focus on single-tier computing architecture comprising only an edge node. To fully exploit the MTC resources, various collaborative edge/fog and Cloud architecture has been proposed in \cite{masip2016foggy, du2017computation, ren2019collaborative, hong2019multi, wang2021multi, wang2022joint}. For example, a hybrid massive MIMO relay and reconfigurable intelligent surface (RIS) aided offloading scheme was conceived in \cite{wang2021multi} to facilitate the exploitation of MTC resources. Another massive-MIMO-relay-aided MTC architecture was studied in \cite{wang2022joint} considering the caching service at the edge node. Furthermore, since communication performance plays an important role in mobile computing, the application of NOMA to facilitate task offloading has also attracted significant attentions \cite{ding2018impact, wang2018multi, song2018joint, yang2020cache, ding2022hybrid}. More specifically, the impact of NOMA on task offloading was evaluated in \cite{ding2018impact} through various asymptotic performance analyses. The authors of \cite{wang2018multi} jointly designed communication and computing resource allocation of NOMA-MEC systems and demonstrated the advantages of NOMA in enhancing energy efficiency. The authors of \cite{song2018joint} and \cite{yang2020cache} further studied and demonstrated the advantages of NOMA in heterogeneous MEC networks and cache-aided MEC networks, respectively. As an innovative contribution, the authors of \cite{ding2022hybrid} conceived a novel hybrid-NOMA offloading scheme for MEC systems, which is appealing in terms of energy efficiency.

\subsubsection{Studies on ISAC}
The design of ISAC systems can be divided into three categories, namely sensing-centric design, communication-centric design, and joint design. For the sensing-centric design, communication can be achieved by embedding information bits into classical sensing signals, such as chirp signals, through beamforming modulation \cite{hassanien2015dual} or index modulation \cite{huang2020majorcom}. On the contrary, the communication-centric design normally relies on the well-known OFDM waveform due to its good properties for both communication and sensing \cite{keskin2021mimo, johnston2022mimo}. However, sensing-centric and communication-centric designs cannot achieve a scalable trade-off between the two functions. As a remedy, joint design is regarded as a promising methodology, which aims at designing the new unified ISAC waveforms \cite{liu2018toward, liu2020beamforming, hua2021optimal}. Recent research contributions have also evaluated the benefits of NOMA in ISAC \cite{mu2022noma, wang2022noma, mu2021noma, wang2022inspire}. The authors of \cite{wang2022noma} proposed a NOMA-empowered ISAC framework, where the superimposed NOMA signals are employed for both communication and sensing to counteract the degradation of spatial DoFs. The authors of \cite{mu2021noma} considered a scenario where the sensing targets also require information transmission and proposed a pair of multicast-unicast ISAC frameworks based on beamformer-based and cluster-based NOMA techniques, respectively. As a further advance, the authors of \cite{wang2022inspire} exploited a composite signal model and developed a NOMA-inspired scheme to achieve the double benefits of the sensing signal for both sensing and communication.

\vspace{-0.3cm}
\subsection{Motivation and Contributions}
Against the above background, the base station (BS) in the future wireless network is expected to provide multiple functionalities, thus realizing new revolutionary applications. As a result, besides the classical wireless communication, it has to provide computing service to the nearby users as well as carrying out target sensing. However, the direct integration of computing and sensing upon the BS may not be a happy marriage. On the one hand, computing and sensing mainly rely on the separated hardware architectures of the BS, namely the computation server and the transceiver. Thus, it is difficult to effectively share the resources between the two functions. On the other hand, the sensing echo signals received by the BS from the target will experience strong interference caused by the computation offloading signals uploaded from the nearby users, namely inter-functionality interference, which significantly degrades the sensing performance. To address these challenges, we propose a NOMA-aided joint communication, sensing, and multi-tier computing (JCSMC) framework to achieve an efficient integration of multiple functionalities at the BS. Firstly, by leveraging the MTC structure, the sensing signal transmitted by the BS can be information-bearing based on the ISAC waveform design and further offload the computation tasks to a Could server (CS). As such, MTC achieves \emph{double benefits} of the sensing signal for both sensing and computing, thus improving  resource efficiency. Secondly, by exploiting NOMA, the computation offloading signals can be decoded through SIC, thus reducing the interference within the computation functionality. Furthermore, the computation offloading signals will not cause interference to the sensing signal since they are already removed in SIC, thus reducing the inter-functionality interference. Hence, NOMA can also achieve \emph{double benefits} for both sensing and computing. The main contributions of this paper are summarized as follows:
\begin{itemize}
    \item We propose a novel NOMA-aided JCSMC framework, in which the MTC and NOMA techniques are employed for facilitating the integration of multiple functionalities. Considering both the partial and binary computation offloading modes, we formulate a computation rate maximization problem for the joint optimization of transmit beamformer and the computing resource allocation, subject to the constraints on communication-computation causality and the sensing quality.

    \item For the partial offloading mode, we propose a weighted minimum mean square error (WMMSE)-based alternating optimization (AO) algorithm to solve the resultant non-convex 
    computation rate maximization problem. More specifically, we first introduce an equivalent communication model of sensing and then establish a set of rate-WMMSE relationships to transform the non-convex problem into a convex one. Furthermore, we theoretically prove that a Karush–Kuhn–Tucker (KKT) optimal solution to the primal problem can be obtained by the proposed algorithm.

    \item For the binary offloading mode, the additional binary variables corresponding to the offloading decisions lead to a highly-coupled mixed-integer optimization problem. To tackle it, we first transform the binary variables and coupled variables into a series of equivalent equality constraints. Then, we propose an alternating direction method of multipliers (ADMM)-based AO algorithm to obtain a high-quality suboptimal solution.

    \item Our simulation results unveil that the proposed NOMA-aided JCSMC framework achieves a better computing-sensing trade-off compared with the benchmark schemes relying on
    the single-tier computing architecture and without NOMA for both partial offloading and binary offloading modes. It also demonstrates that NOMA plays an important role in guaranteeing computing performance when the requirement of sensing quality and user number increase. 
\end{itemize}

\subsection{Organization and Notations}

The rest of this paper is organized as follows. In Section \ref{sec:system_model}, the NOMA-aided JCSMC framework is conceived. In Section \ref{sec:partial_offloading}, a computation rate maximization problem is formulated for the partial offloading mode, which is solved by the developed WMMSE-based AO algorithm. In Section \ref{sec:binary_offloading}, another computation rate maximization problem is formulated for the binary offloading mode, which is addressed by the proposed ADMM-based AO algorithm. In Section \ref{sec:results}, the numerical results are provided for characterizing the proposed framework compared with the benchmark schemes. Finally, Section \ref{sec:conclusion} concludes this paper.

\emph{Notations}: Scalars, vectors, and matrices are denoted by the lower-case, bold-face lower-case, and bold-face upper-case letters, respectively; $\mathbb{C}^{N \times M}$ and $\mathbb{R}_+^{N \times M}$ denotes the space of $N \times M$ complex and non-negative real matrices, respectively; $|a|$ denotes the magnitude of scalar $a$; $\mathbf{a}^H$ and $\|\mathbf{a}\|_p$ denotes the conjugate transpose and $\ell_p$-norm of vector $\mathbf{a}$;  $\mathrm{diag}(\mathbf{a})$ denotes a diagonal matrix with same value as the vector $\mathbf{a}$ on the diagonal; $\mathrm{tr}(\mathbf{A})$ denotes the trace of matrix $\mathbf{A}$; $\mathbb{E}[\cdot]$ denotes the statistical expectation and $\mathrm{Re}\{\cdot\}$ denotes the real component of a complex number; $\mathcal{CN}(\mu, \sigma^2)$ denotes the distribution of a circularly symmetric complex Gaussian (CSCG) random variable with mean $\mu$ and variance $\sigma^2$.

\section{System Model} \label{sec:system_model}

As shown in Fig. \ref{fig:system_model}, we propose a NOMA-aided JCSMC framework comprising an $N$-antenna multi-functional BS, $K$ single-antenna computation users, whose indices are collected in $\mathcal{K}=\{1,\dots,K\}$, one sensing target under an angle of $\theta_0$, $M$ clutter sources under angles of $\{\theta_m\}_{m=1}^M$. Moreover, the computation tasks at the BS can also be further offloaded to a single-antenna CS by leveraging the joint ISAC signal. We assume that the users always offload the computation tasks to the BS due to their limited computing capabilities. Furthermore, it is also assumed that there is no direct link between the users and the CS \cite{hong2019multi, wang2022joint}.

\begin{figure}[t!]
    \centering
    \includegraphics[width=0.4\textwidth]{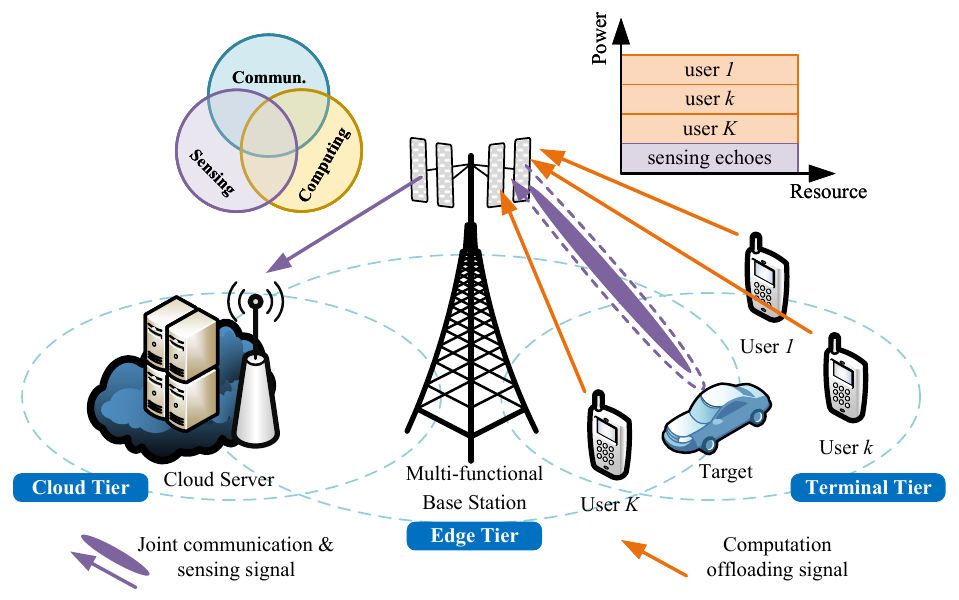}
    \caption{Illustration of the proposed NOMA-aided JCSMC framework.}
    \label{fig:system_model}
\end{figure}

\subsection{Signal Model}
In the proposed JCSMC framework, the received signal at the multi-functional BS consists of three parts, namely the computation offloading signal from $K$ computation users, echo signal from one desired sensing target at the angle of $\theta_0$, and clutter signals from $M$ undesired clutter sources. In this paper, we assume that the BS has prior knowledge of the target and clutter sources. In particular, the prior knowledge of the target can be acquired based on the estimation results in the previous sensing phase, while that of the clutter sources can be obtained from a cognitive paradigm based on the environmental dynamic database \cite{lane2010cognitive, aubry2013knowledge, liu2021joint}. Moreover, the channel state information (CSI) for communication is also assumed to be perfectly known at the BS. Denote $\tilde{s}_{u,k} = \sqrt{P_u} s_{u,k}$ as the transmit signal of computation user $k$ for delivering the information stream $s_{u,k} \in \mathbb{C}$ to the BS with transmit power $P_u$. Then, the received superimposed computation offloading signals from $K$ users at the BS is given by 
\begin{equation}
    \mathbf{y}_c = \sum_{k \in \mathcal{K}} \sqrt{P_u} \mathbf{h}_{u,k} s_{u,k},
\end{equation}
where $\mathbf{h}_{u,k} = L_{u,k}^{-1/2} \tilde{\mathbf{h}}_{u,k}$ denotes the channel between the BS and computation user $k$, $L_{u,k} \in \mathbb{R}_+$ and $\tilde{\mathbf{h}}_{u,k} \in \mathbb{C}^{N \times 1}$ denote the path loss and small-scale fading, respectively. Furthermore, denote $\mathbf{x}$ as the transmit signal at the BS for target sensing. Then, the echo signal from the desired target plus the clutter signals received at the BS can be written as 
\begin{equation}
    \mathbf{y}_s = \underbrace{\beta_0 \mathbf{a}(\theta_0) \mathbf{a}^T(\theta_0) \mathbf{x}}_{ \text{echo signal from desired target}} + \underbrace{\sum_{m \in \mathcal{M} \setminus 0} \beta_m \mathbf{a}(\theta_m) \mathbf{a}^T(\theta_m) \mathbf{x}}_{\scriptstyle \text{clutter signals from clutter sources}},
\end{equation}
where $\mathcal{M} = \{0,\dots,M\}$ denotes the index set of the desired target and clutter sources, $\beta_m = L_{s,m}^{-1/2} \tilde{\beta}_m$ includes both round-trip path loss $L_{s,m} \in \mathbb{R}_+$ and the complex reflection factor $\tilde{\beta}_m \in \mathbb{C}$, and $\mathbf{a}(\theta_m) \in \mathbb{C}^{N \times 1}$ denotes the steering vector of the antenna array at the BS under the angle of $\theta_m$. Assuming the planar wave and the uniform linear array (ULA) at the BS, the steering vector is given by 
\begin{equation}
    \mathbf{a}(\theta_m) = [1, e^{j\frac{2\pi}{\lambda}d\sin({\theta_m})},...,e^{j\frac{2\pi}{\lambda}d(N-1)\sin({\theta_m})}]^T,
\end{equation}
where $\lambda$ and $d$ denote the signal wavelength and antenna spacing, respectively. Note that the BS in the considered system operates as mono-static sensing, which can be subject to strong self-interference (SI). Nevertheless, with the advanced SI cancellation techniques, a $100$dB SI suppression can be achieved \cite{liu2021integrated}. Accordingly, we assume the perfect SI cancellation in this paper. As a consequence, the overall received signal at the BS is given by 
\begin{equation} \label{eqn:BS_signal}
    \mathbf{y}_u = \mathbf{y}_c + \mathbf{y}_s + \mathbf{n}_u,
\end{equation}
where and $\mathbf{n}_u \sim \mathcal{CN}(\mathbf{0}, \sigma_u^2 \mathbf{I}_N)$ denotes the complex Gaussian noise at the receiver. 

In order to fully exploit the transmit signal $\mathbf{x}$ at the BS, the information stream is embedded into the signal for further offloading the computation task to the powerful CS. In this case, the transmit signal $\mathbf{x}$ can be rewritten as 
\begin{equation}
    \mathbf{x} = \mathbf{p} s_d,
\end{equation} 
where $\mathbf{p} \in \mathbb{C}^{N \times 1}$ denotes the transmit beamformer and $s_d \in \mathbb{C}$ denotes the information stream from the BS to the CS. Hence, the received signal at the CS is given by 
\begin{equation} \label{eqn:receive_FN}
    y_d = \mathbf{h}_d^H \mathbf{x} + n_d = \mathbf{h}_d^H \mathbf{p} s_d + n_d,
\end{equation}
where $\mathbf{h}_d \in \mathbb{C}^{N \times 1}$ denotes the BS-CS channel and $n_d \sim \mathcal{CN}(0, \sigma_d^2)$ denote the complex Gaussian noise. 
We also assume that $s_d$ and $\{s_{u,k}\}_{k=1}^K$ are independent random variables with zero mean and unit power. Note that even though the sensing signal is information-bearing, the sensing functionality can still be carried out since the full knowledge of the transmit waveform is still available at the BS.

\subsection{Communication Model} \label{sec:communication_model}
The performance of the computation task offloading in the proposed NOMA-aided JCSMC framework is determined by the achievable communication rate. For decoding the computation offloading signal from $K$ computation users at the BS, the uplink NOMA is exploited, where the receiver exploits the SIC to successively decode each stream and cancel the effect of the decoded streams from the received signal \cite{liu2022evolution}. Let the permutation $\pi$ over the set $\mathcal{K}$ denote the successive decoding order of the streams from $K$ computation users. For instance, if $\pi(k) = j$, $s_{u,k}$ will be the $j$-th stream to be decoded in the SIC process. With this decoding order, all the information streams $s_{u,i}$ with $\pi(i) < \pi(k)$ are decoded and cancelled from the received signal before decoding $s_{u,k}$. Thus, the effective signal for decoding $s_{u,k}$ via the receiver $\mathbf{w}_{u,k} \in \mathbb{C}^{N \times 1}$ is given by 
\begin{align} \label{eqn:NOMA_signal}
    &\mathbf{w}_{u,k}^H \mathbf{y}_{u,k} = \sqrt{P_u} \mathbf{w}_{u,k}^H \mathbf{h}_{u,k} s_{u,k} \nonumber \\
    &+ \mathbf{w}_{u,k}^H \underbrace{\left(\sum_{ \pi(i) > \pi(k) } \sqrt{P_u} \mathbf{h}_{u,i} s_{u,i} + (\mathbf{H}_s + \mathbf{H}_c) \mathbf{p}s_d + \mathbf{n}_u\right)}_{\text{inter-user interference plus noise }\tilde{\mathbf{n}}_k},
\end{align} 
where $\mathbf{H}_s = \beta_0 \mathbf{a}(\theta_0) \mathbf{a}(\theta_0)^T$ and $\mathbf{H}_c = \sum_{m=1}^M \beta_m \mathbf{a}(\theta_m) \mathbf{a}(\theta_m)^T$. Then, the achievable communication rate (in bit/s/Hz) from user $k$ to the BS is given by 
\begin{equation} \label{eqn:rate_user_EN}
    R_{u,k} = \log_2\left( 1 + \frac{P_u |\mathbf{w}_{u,k}^H \mathbf{h}_{u,k}|^2 }{\mathbf{w}_{u,k}^H \mathbf{R}_{\tilde{\mathbf{n}}_{k}} \mathbf{w}_{u,k} } \right).
\end{equation}
where the matrix $\mathbf{R}_{\tilde{\mathbf{n}}_{k}}$ denotes the covariance matrix of the inter-user interference plus noise $\tilde{\mathbf{n}}_{k}$ as follows:
\begin{align}
    &\mathbf{R}_{\tilde{\mathbf{n}}_{k}} = \mathbb{E} [\tilde{\mathbf{n}}_{k} \tilde{\mathbf{n}}_{k}^H ] \nonumber \\
    = & \!\! \sum_{\pi(i)>\pi(k)} \!\! P_u \mathbf{h}_{u,i} \mathbf{h}_{u,i}^H + (\mathbf{H}_s+\mathbf{H}_c) \mathbf{p} \mathbf{p}^H (\mathbf{H}_s+\mathbf{H}_c)^H + \sigma_u^2 \mathbf{I}_N.
\end{align}
According to \cite[Chapter 8]{tse2005fundamentals}, MMSE receiver is optimal to achieve the communication capacity for uplink NOMA, which is designed as follows:
\begin{equation}
    \mathbf{w}_{u,k}^{\text{MMSE}} = \sqrt{P_u} \mathbf{R}_{\tilde{\mathbf{n}}_{k}}^{-1} \mathbf{h}_{u,k}.
\end{equation} 
By substituting $\mathbf{w}_{u,k}^{\text{MMSE}}$ into \eqref{eqn:rate_user_EN}, the achievable communication rate can be rewritten as
\begin{equation}
    R_{u,k} = \log_2 \left( 1 + P_u \mathbf{h}_{u,k}^H \mathbf{R}_{\tilde{\mathbf{n}}_{k}}^{-1}  \mathbf{h}_{u,k} \right).
\end{equation} 
Furthermore, according to the received signal at the CS in \eqref{eqn:receive_FN}, the achievable communication rate (in bit/s/Hz) from the BS to the CS is given by 
\begin{equation}
    R_d = \log_2 \left( 1 + \frac{|\mathbf{h}_d^H \mathbf{p}|^2}{\sigma_d^2} \right).
\end{equation}

\subsection{Sensing Model}
After decoding all the offloading signals from $K$ users via SIC, there will be no interference from the offloading signals to the target sensing, i.e., the signal $\mathbf{y}_c$ has been removed from $\mathbf{y}_u$ in \eqref{eqn:BS_signal}. Thus, the effective signal via the receiver $\mathbf{w}_s \in \mathbb{C}^{N \times 1}$ for target sensing is given by 
\begin{equation} \label{eqn:effective_sensing_signal}
    \mathbf{w}_s^H (\mathbf{y}_u - \mathbf{y}_c) = \mathbf{w}_s^H \mathbf{H}_s \mathbf{p} s_d + \mathbf{w}_s^H \underbrace{\left( \mathbf{H}_c \mathbf{p} s_d + \mathbf{n}_u \right)}_{\text{clutter interference plus noise } \mathbf{c}},
\end{equation} 
which yields the following sensing SINR \cite{chen2009mimo}:
\begin{equation} \label{eqn:sensing_SINR_0}
    \gamma_s = \frac{ |\mathbf{w}_s^H \mathbf{H}_s \mathbf{p}|^2 }{ \mathbf{w}_s^H \mathbf{R}_{\mathbf{c}} \mathbf{w}_s}.
\end{equation}
Here, $\mathbf{R}_{\mathbf{c}}$ denotes the covariance matrix of the clutter interference plus noise, which is given by
\begin{equation}
    \mathbf{R}_{\mathbf{c}} = \mathbb{E}[\mathbf{c} \mathbf{c}^H] = \mathbf{H}_c \mathbf{p} \mathbf{p}^H \mathbf{H}_c^H + \sigma_u^2 \mathbf{I}_N.
\end{equation}
Note that the sensing SINR $\gamma_s$ determines the performance of both detection and localization of the target sensing \cite{chen2022generalized}. Thus, in order to guarantee the sensing performance, the sensing SINR $\gamma_s$ should satisfy the minimum value constraint, i.e., $\gamma_s \ge \gamma_{r, \min}$. In this case, the optimal receiver $\mathbf{w}_s^\star$ needs to be designed for maximizing $\gamma_s$, which can be transformed into the well-known minimum variance distortionless response (MVDR) problem. According to \cite{chen2009mimo}, the optimal receiver is given by 
\begin{equation} \label{eqn:sensing_rev}
    \mathbf{w}_s^{\text{MVDR}} = \arg \max_{\mathbf{w}_s} \gamma_s = \frac{\mathbf{R}_{\mathbf{c}}^{-1}\mathbf{H}_s\mathbf{p}}{\mathbf{p}^H \mathbf{H}_s^H \mathbf{R}_{\mathbf{c}}^{-1}\mathbf{H}_s\mathbf{p} }.
\end{equation}
By substituting $\mathbf{w}_s^{\text{MVDR}}$ into \eqref{eqn:sensing_SINR_0}, the sensing SINR $\gamma_s$ can be rewritten as
\begin{equation} \label{eqn:sensing_SINR}
    \gamma_s = \mathbf{p}^H \mathbf{H}_s^H \mathbf{R}_{\mathbf{c}}^{-1} \mathbf{H}_s \mathbf{p}.
\end{equation}
In the above two subsections, we first successively decode all the communication signal $\mathbf{y}_u$ subject to the interference from the sensing via SIC at the BS, and the sensing signal $\mathbf{y}_s$ is communication-interference free. Indeed, since the sensing signal $\mathbf{y}_s$ is information-bearing in the proposed framework, it can also be decoded firstly subject to communication interference, and the communication signal will be sensing-interference free. However, such an operation mechanism is not preferred in the proposed framework due to the following reasons. On the one hand, since the sensing signal experiences a round-trip pathloss, the power of $\mathbf{y}_s$ is typically much lower than that of $\mathbf{y}_u$, resulting in a very low SINR for decoding $\mathbf{y}_s$ in SIC. On the other hand, the embedded information stream $s_d$ need to be decodable at both BS and CS to carry out SIC. In this case, the achievable rate of $s_d$ at the CS will be significantly limited by that at the BS.

\subsection{Computation Model}
Next, we present the computation model of the proposed NOMA-aided JCSMC framework, where the delay caused by downloading the computation results from the servers is omitted due to the superior computing capability of the servers and the relatively small size of the results. Denote $\mathcal{X}_{\mathrm{BS}} \subseteq \mathcal{K}$ and $\mathcal{X}_{\mathrm{CS}} \subseteq \mathcal{K}$ as the sets of users that execute their computation tasks at the BS and the CS, respectively, $\phi$ as the computation cycles for processing one bit of data from each user, which depends on the nature of the computation tasks, and $f_{b,k}$ as the CPU-cycle frequency (cycle/s) allocated by the BS to user $k$. Then, the computation rate, i.e., the number of bits processed per second, of user $k$ at the BS is given by
\begin{equation} \label{eqn:comp_rate_edge}
    r_{b,k} = \frac{f_{b,k}}{\phi}, \forall k \in \mathcal{X}_{\mathrm{BS}} \quad \text{and} \quad r_{b,k} = 0, \forall k \notin \mathcal{X}_{\mathrm{BS}}.
\end{equation}
Similarly, the computation rate of user $k$ at the CS is given by
\begin{equation} \label{eqn:comp_rate_cloud}
    r_{c,k} = \frac{f_{c,k}}{\phi}, \forall k \in \mathcal{X}_{\mathrm{CS}} \quad \text{and} \quad r_{c,k} = 0, \forall k \notin \mathcal{X}_{\mathrm{CS}},
\end{equation} 
where $f_{c,k}$ denotes the CPU-cycle frequency allocated to user $k$ at the CS.
Since the BS is normally power constrained, the power consumption for computing at the BS should be considered, which is given by \cite{mao2017survey}
\begin{equation}
    p_b = \sum_{k \in \mathcal{X}_{\mathrm{BS}}} \kappa f_{b,k}^3 = \sum_{k \in \mathcal{X}_{\mathrm{BS}}} \kappa (\phi r_{b,k})^3,
\end{equation}
where $\kappa$ is the power factor related to the CPU architecture of the computation server at the BS. 
Furthermore, due to the causality relationship between computing and task offloading, the computation rate $r_{b,k}$ and $r_{c,k}$ need to satisfy the following causality constraints:
\begin{subequations}
    \begin{align}
        \label{causality_1}
        & r_{b,k} + B R_{d,k} \le B R_{u, k}, \forall k \in \mathcal{K}\\
        \label{causality_2}
        & \sum_{k \in \mathcal{K}}  R_{d,k} \le R_d,\\
        \label{causality_3}
        & r_{c,k} \le B R_{d,k}, \forall k \in \mathcal{K}, 
    \end{align}
\end{subequations}
where $B$ denotes the channel bandwidth and $R_{d,k}$ denotes the offloading rate from the BS to the CS for user $k$. More specifically, \eqref{causality_1} indicates that for user $k$, the computation rate at the BS plus the offloading rate from the BS to the CS cannot exceed the communication rate from user $k$ to the BS. \eqref{causality_2} ensures that the overall offloading rate of $K$ users from the BS to the CS is upper bounded by the available communication rate. \eqref{causality_3} represents that the computation rate for user $k$ at the CS should be smaller than the corresponding computation offloading rate. Generally, the equality in \eqref{causality_3} can be achieved since the CS is normally with sufficient power for supporting high-performance computing. Thus, the causality constraints can be simplified as the following communication-computation causality constraints:
\begin{subequations}  \label{eqn:rate_upper_bound}
    \begin{align} 
        & r_{b,k} + r_{c,k} \le B R_{u, k}, \forall k \in \mathcal{K}, \\
        & \sum_{k \in \mathcal{K}} r_{c,k}\le B R_d.
    \end{align}
\end{subequations}
Then, the overall computation rate is given by 
\begin{equation}
    \tilde{R}= \sum_{k \in \mathcal{K}} \omega_k (r_{b,k} + r_{c,k})
\end{equation}
where $\boldsymbol{\omega}=[\omega_1,...,\omega_k]^T \in \mathbb{R}^{K \times 1}_+ $ is the weights for characterizing the priority of each user. 

Finally, two computation offloading modes are considered for the computation model, namely partial offloading and binary offloading \cite{mao2017survey}, to define the user sets  $\mathcal{X}_{\mathrm{BS}}$ and $\mathcal{X}_{\mathrm{CS}}$. For the partial offloading mode, the task bits can be arbitrarily divided into two groups and executed at the BS and the CS, respectively. In this case, it holds that $\mathcal{X}_{\mathrm{BS}} = \mathcal{X}_{\mathrm{CS}} = \mathcal{K}$. However, in practice, the computation tasks may be highly integrated and cannot be partitioned, for which the binary offloading mode is needed. For the binary offloading mode, the computation task of each user can only be executed as a whole at either the BS or the CS, and it holds that $\mathcal{X}_{\mathrm{BS}} \subseteq \mathcal{K}$ and $\mathcal{X}_{\mathrm{CS}} = \mathcal{K} \backslash \mathcal{X}_{\mathrm{BS}}$.  

\begin{remark} \label{remark_1}
    \emph{
    \emph{(Computing-sensing trade-off)} There is a trade-off between computing and sensing performance in the JCSMC framework. Firstly, for achieving higher sensing SINR, more power needs to be allocated to the beamformer $\mathbf{p}$, resulting in the degradation of the computation rate at the BS because of the less power available for computing. Secondly, similar to pure ISAC systems, more power of the BS transmitter will be concentrated at the angle of the target to achieve better sensing performance. In this case, the BS will transmit less power toward the CS, leading to a lower offloading rate and therefore a lower computation rate.
    }
\end{remark}

\begin{remark} \label{remark_2}
    \emph{
    \emph{(Benefits of NOMA)} The main benefits of NOMA to the JCSMC framework are summarized as follows. Firstly, NOMA is the capacity-achieving scheme for uplink transmission in the multi-antenna systems \cite{tse2005fundamentals}, thus resulting in the maximum computation rate upper bound. Secondly, by exploiting the SIC technique, the effective signal for target sensing is free from the interference of the offloading signals, thus reducing the inter-functionalities interference and achieving a high-quality sensing performance, see \eqref{eqn:effective_sensing_signal}. Thirdly, based on the aforementioned benefits and analysis in \textbf{Remark \ref*{remark_1}}, NOMA is capable of enlarging the computing-sensing trade-off region.
    }
\end{remark}

\section{Partial Computation Offloading} \label{sec:partial_offloading}

In this section, we first focus on the partial computation offloading mode. Our goal is to jointly optimize the computing resource allocation $\{r_{b,k}\}$ and $\{r_{c,k}\}$ as well as the transmit beamformer $\mathbf{p}$ to maximize the computation rate $\tilde{R}$, while satisfying the constraints of communication-computation causality and sensing quality. To solve the resultant non-convex optimization problem, a WMMSE-based algorithm is proposed to obtain a KKT optimal solution.

\subsection{Problem Formulation}
For the partial computation offloading, where $\mathcal{X}_{\mathrm{BS}} = \mathcal{X}_{\mathrm{CS}} = \mathcal{K}$, the computation rate $r_{b,k}$ and $r_{c,k}$ become $r_{b,k} = \frac{f_{b,k}}{\phi}, \forall k \in \mathcal{K}$ and $r_{c,k} = \frac{f_{c,k}}{\phi}, \forall k \in \mathcal{K}$, respectively. Thus, the computation rate maximization problem with a given decoding order $\pi$ can be formulated as follows:
\begin{subequations}
    \begin{align}
        \text{(P1):} \quad \max_{\mathbf{r}_b, \mathbf{r}_c, \mathbf{p}} \quad &  \sum_{k \in \mathcal{K}} \omega_k (r_{b,k} + r_{c,k}) \\
        \label{constraint:causality_1}
        \mathrm{s.t.} \quad & r_{b,k} + r_{c,k} \le B R_{u, k}, \forall k \in \mathcal{K}, \\
        \label{constraint:causality_2}
        & \sum_{k \in \mathcal{K}} r_{c,k}\le B R_d, \\
        \label{constraint:sensing_SINR}
        & \gamma_s \ge \gamma_{s, \min}, \\
        \label{constraint:resource_compute}
        & \sum_{k \in \mathcal{K}} \kappa (\phi r_{b,k})^3 + \mathrm{tr}(\mathbf{p}\mathbf{p}^H) \le P_b,\\
        \label{constraint:positive}
        & r_{b,k} \ge 0, r_{b,k} \ge 0, \forall k \in \mathcal{K},
    \end{align}
\end{subequations}
where $\mathbf{r}_b = [r_{b,1},\dots,r_{b,K}]^T$ and $\mathbf{r}_c = [r_{c,1},\dots,r_{c,K}]^T$ denote the vector of the computing rate at the BS and the CS, respectively, 
and $P_b$ in \eqref{constraint:resource_compute} denotes the total power available at the BS. Constraints \eqref{constraint:causality_1} and \eqref{constraint:causality_2} represent the causality between communication and computation. Constraint \eqref{constraint:sensing_SINR} guarantees the minimum quality of the target sensing. However, (P1) is a highly-coupled non-convex problem due to the constraints \eqref{constraint:causality_1}-\eqref{constraint:sensing_SINR}, which is generally challenging to obtain the globally optimal solution.

\subsection{Equivalent Communication Model for Target Sensing}
In order to tackle the sensing SINR constraint \eqref{constraint:sensing_SINR} in (P1), we first transform it into a homogeneous form as the constraints \eqref{constraint:causality_1} and \eqref{constraint:causality_2} by introducing an equivalent communication model for sensing, which is given by 
\begin{equation}
  \tilde{\mathbf{y}}_s = \mathbf{H}_s \mathbf{p} s_s + \mathbf{z}_s,
\end{equation}
where $s_s \in \mathbb{C}$ is the virtual communication stream and $\mathbf{z}_s \in \mathbb{C}^{N \times 1}$ is the virtual colored noise with the covariance matrix $\mathbf{R}_{\mathbf{c}}$. Following the same path in Section \ref{sec:communication_model}, the achievable communication rate for decoding $s_s$ via the MMSE receiver is given by 
\begin{equation}
  R_s = \log_2 \left(1 + \mathbf{p}^H \mathbf{H}_s^H \mathbf{R}_{\mathbf{c}}^{-1}\mathbf{H}_s\mathbf{p} \right).
\end{equation} 
Then, it follows that $R_s = \log_2(1 + \gamma_s)$. Since the function $\log_2(\cdot)$ is monotonically increasing, problem (P1) can be rewritten as 
\begin{subequations}
  \begin{align}
      \text{(P1.1):} \quad \max_{\mathbf{r}_b, \mathbf{r}_c, \mathbf{p}} \quad &  \sum_{k \in \mathcal{K}} \omega_k (r_{b,k} + r_{c,k}) \\
      \label{constraint:1.1.1}
      \mathrm{s.t.} \quad & r_{b,k} + r_{c,k} \le B R_{u, k}, \forall k \in \mathcal{K}, \\
      \label{constraint:1.1.2}
      & \sum_{k \in \mathcal{K}} r_{c,k}\le B R_d, \\
      \label{constraint:1.1.3}
      & \Gamma_{s, \min} \le R_s, \\
      & \eqref{constraint:resource_compute}, \eqref{constraint:positive},
  \end{align}
\end{subequations}
where $\Gamma_{s, \min} = \log_2(1 + \gamma_{s, \min})$. 

\subsection{WMMSE-based AO Algorithm for Solving (P1.1)} \label{sec:partial_wmmse}
To solve the non-convex constraints \eqref{constraint:1.1.1}-\eqref{constraint:1.1.3}, the WMMSE framework is invoked. Denote $\tilde{\mathbf{w}}_{u,k}$, $\tilde{w}_d$, and $\tilde{\mathbf{w}}_s$ as the receiver for decoding the desired stream from the signals $\mathbf{y}_{u,k}$, $y_d$, and $\tilde{\mathbf{y}}_s$. Then, the mean square error (MSE) for decoding the desired streams with these receivers is given by 
\begin{align} 
  \label{eqn:up_MSE}
  &e_{u,k} = \mathbb{E} \left[| \tilde{\mathbf{w}}_{u,k}^H \mathbf{y}_{u,k} - s_{u,k} |^2\right] \nonumber \\
  = &\tilde{\mathbf{w}}_{u,k}^H \! \left( P_t\mathbf{h}_{u,k} \mathbf{h}_{u,k}^H \!+\! \mathbf{R}_{\tilde{\mathbf{n}}_k} \right) \!\tilde{\mathbf{w}}_{u,k} \!-\! 2\mathrm{Re}\{\sqrt{P_t} \tilde{\mathbf{w}}_{u,k}^H \!\mathbf{h}_{u,k}\} \!+\! 1, \\
  \label{eqn:down_MSE}
  &e_d= \mathbb{E}\left[| \tilde{w}_d^* y_d - s_d |^2\right] \nonumber \\
  = &|\tilde{w}_d|^2\left( |\mathbf{h}_d^H \mathbf{p}|^2 + \sigma_d^2 \right) - 2 \mathrm{Re} \{ \tilde{w}_d^* \mathbf{h}_d^H \mathbf{p}  \} + 1,\\
  \label{eqn:sensing_MSE}
  &e_s = \mathbb{E} \left[ | \tilde{\mathbf{w}}_s^H \tilde{\mathbf{y}}_s - s_s |^2 \right] \nonumber \\
  = &\tilde{\mathbf{w}}_s^H \left( \mathbf{H}_s\mathbf{p} \mathbf{p}^H \mathbf{H}_s^H + \mathbf{R}_{\mathbf{c}} \right)\tilde{\mathbf{w}}_s - 2 \mathrm{Re} \{\tilde{\mathbf{w}}_s^H \mathbf{H}_s\mathbf{p}\} + 1.
\end{align} 
Then, the optimal receiver $\tilde{\mathbf{w}}_{u,k}$ for minimizing the MSE $e_{u,k}$ is given by \cite[Chapter 8]{tse2005fundamentals}
\begin{equation} \label{eqn:mmse_seceiver_1}
  \tilde{\mathbf{w}}_{u,k}^{\text{MMSE}} = \arg \min_{\tilde{\mathbf{w}}_{u,k}} e_{u,k} = \sqrt{P_t}\left( P_t \mathbf{h}_k \mathbf{h}^H_k + \mathbf{R}_{\tilde{\mathbf{n}}_k}\right)^{-1} \mathbf{h}_k,
\end{equation}
which can be obtained by solving $\nabla_{\tilde{\mathbf{w}}_{u,k}} e_{u,k} = 0$. Similarly, the optimal receivers $\tilde{w}_d$ and $\tilde{\mathbf{w}}_s$ for minimizing the MSEs $e_d$ and $e_s$ are given by 
\begin{align}
  \label{eqn:mmse_seceiver_2}
  &\tilde{w}_d^{\text{MMSE}} = \left(|\mathbf{h}_d^H \mathbf{p}|^2 + \sigma_d^2 \right)^{-1} \mathbf{h}_d^H \mathbf{p},\\
  \label{eqn:mmse_seceiver_3}
  &\tilde{\mathbf{w}}_s^{\text{MMSE}} = \left( \mathbf{H}_s\mathbf{p} \mathbf{p}^H \mathbf{H}_s^H + \mathbf{R}_{\mathbf{c}} \right)^{-1} \mathbf{H}_s\mathbf{p},
\end{align} 
By substituting the optimal receivers $\tilde{\mathbf{w}}_{u,k}^{\text{MMSE}}$, $\tilde{w}_d^{\text{MMSE}}$, and $\tilde{\mathbf{w}}_s^{\text{MMSE}}$ into \eqref{eqn:up_MSE}, \eqref{eqn:down_MSE}, and \eqref{eqn:sensing_MSE}, respectively, the corresponding minimum MSEs can be obtained as follows:
\begin{align}
  e^{\text{MMSE}}_{u,k} &= 1 - P_t \mathbf{h}_{u,k}^H (P_t \mathbf{h}_{u,k} \mathbf{h}_{u,k}^H  + \mathbf{R}_{\tilde{\mathbf{n}}_k})^{-1} \mathbf{h}_{u,k} \nonumber\\
  &\overset{(a)}{=} \left(1 + P_t \mathbf{h}_{u,k}^H \mathbf{R}_{\tilde{\mathbf{n}}_k}^{-1} \mathbf{h}_{u,k} \right)^{-1}, \\
  e^{\text{MMSE}}_d &= 1 - \left(|\mathbf{h}_d^H \mathbf{p}|^2 + \sigma_d^2 \right)^{-1} |\mathbf{h}_d^H \mathbf{p}|^2 \nonumber\\
  &= \left( 1 + |\mathbf{h}_d^H \mathbf{p}|^2/\sigma_d^2 \right)^{-1},\\
  e^{\text{MMSE}}_s &= 1 - \mathbf{p}^H\mathbf{H}_s^H \left( \mathbf{H}_s\mathbf{p} \mathbf{p}^H \mathbf{H}_s^H + \mathbf{R}_{\mathbf{c}} \right)^{-1} \mathbf{H}_s\mathbf{p} \nonumber \\
  &\overset{(b)}{=} \left( 1 + \mathbf{p}^H\mathbf{H}_s^H \mathbf{R}_{\mathbf{c}}^{-1} \mathbf{H}_s\mathbf{p} \right)^{-1}, 
\end{align}
where $(a)$ and $(b)$ are obtained by applying the Woodbury matrix identity \cite{higham2002accuracy}. Therefore, the communication rate in the constraints \eqref{constraint:1.1.1}-\eqref{constraint:1.1.3} can be rewritten as $R_{u,k} = -\log_2 \left(e^{\text{MMSE}}_{u,k}\right)$, $R_d = -\log_2 \left( e^{\text{MMSE}}_d \right)$, and $R_s = -\log_2 \left( e^{\text{MMSE}}_s \right)$. Next, we introduce the positive weights $\varpi_{u,k}$, $\varpi_d$, and $\varpi_s$, and define the augmented weighted MSEs (AWMSEs) as $\varepsilon_{u,k} = \varpi_{u,k} e_{u,k} - \log_2 \varpi_{u,k}$, $\varepsilon_d = \varpi_d e_d - \log_2\varpi_d$, and $\varepsilon_s = \varpi_s e_s - \log_2 \varpi_s$.    
Based on the AWMSEs, the following rate-WMMSE relationships can be established: 
\begin{align}
  \label{eqn:comm_rate_wmmse}
  &\varepsilon_{u,k}^{\text{MMSE}} = \min_{\varpi_{u,k}, \tilde{\mathbf{w}}_{u,k}} \varepsilon_{u,k} = c - R_{u,k},\\
  &\varepsilon_d^{\text{MMSE}} = \min_{\varpi_d, \tilde{w}_d} \varepsilon_d = c - R_d,\\
  &\varepsilon_s^{\text{MMSE}} = \min_{\varpi_s, \tilde{\mathbf{w}}_s} \varepsilon_s = c - R_s,
\end{align}
where $c = 1/\ln 2 + \log_2 \ln2$. The above relationships can be obtained by substituting the optimal receivers and weights back to the AWMSEs. For instance, the optimal receiver $\tilde{\mathbf{w}}_{u,k}$ for \eqref{eqn:comm_rate_wmmse} is $\tilde{\mathbf{w}}_{u,k}^\star = \tilde{\mathbf{w}}_{u,k}^{\text{MMSE}}$, which can be derived by $\nabla_{\tilde{\mathbf{w}}_{u,k}} \varepsilon_{u,k} = 0$. Given the optimal receiver $\tilde{\mathbf{w}}_{u,k}^\star$, the corresponding optimal weight $\varpi_{u,k}$ can be calculated by $\frac{\partial \varepsilon_k(\tilde{\mathbf{w}}_{u,k}^\star)}{\partial \varpi_{u,k}}=0$, yielding
\begin{equation} \label{eqn:mmse_weight_1}
  \varpi_{u,k}^\star = \varpi_{u,k}^{\text{MMSE}} \triangleq \left(e^{\text{MMSE}}_{u,k} \ln 2 \right)^{-1}.
\end{equation}
Similarly, the optimal receivers 
$\tilde{w}_d$ and $\tilde{\mathbf{w}}_s$ are $\tilde{w}_d^\star = \tilde{w}_d^{\text{MMSE}}$ and $\tilde{\mathbf{w}}_s^\star = \tilde{\mathbf{w}}_s^{\text{MMSE}}$, respectively. The corresponding optimal weights are given by 
\begin{align}
  \label{eqn:mmse_weight_2}
  &\varpi_d^\star = \varpi_d^{\text{MMSE}} \triangleq \left(e^{\text{MMSE}}_d \ln 2 \right)^{-1}, \\
  \label{eqn:mmse_weight_3}
  &\varpi_s^\star = \varpi_s^{\text{MMSE}} \triangleq \left( e_s^{\text{MMSE}} \ln 2 \right)^{-1}.
\end{align}
By applying the rate-WMMSE relationships, we can reformulate the problem (P1.1) as follows:
\begin{subequations}
  \begin{align}
      \text{(P1.2):} \quad &\max_{ \scriptstyle \mathbf{r}_b, \mathbf{r}_c, \mathbf{p}, \atop \scriptstyle \boldsymbol{\varpi}, \tilde{\mathbf{w}}} \quad   \sum_{k \in \mathcal{K}} \omega_k (r_{b,k} + r_{c,k}) \\
      \label{constraint:up_rate_mmse}
      \mathrm{s.t.} \quad & r_{b,k} + r_{c,k} \le B c - B \varepsilon_{u,k}, \forall k \in \mathcal{K}, \\
      \label{constraint:down_rate_mmse}
      & \sum_{k \in \mathcal{K}} r_{c,k}\le B c - B \varepsilon_d, \\
      \label{constraint:sensing_mmse}
      & \Gamma_{r,\min} \le c - \varepsilon_s, \\
      & \eqref{constraint:resource_compute},\eqref{constraint:positive},
  \end{align}
\end{subequations}
where $\boldsymbol{\varpi} = [\varpi_{u,1},\dots,\varpi_{u,k}, \varpi_s, \varpi_d]^T$ denotes the vector of weights of AWMSEs and $\tilde{\mathbf{w}} = [\tilde{\mathbf{w}}_{u,1}^T,$ $\dots,\tilde{\mathbf{w}}_{u,K}^T, \tilde{\mathbf{w}}_s^T, \tilde{w}_d]^T$ denotes the vector of receivers. Regarding the optimality relationship between problems (P1.1) and (P1.2), we have the following proposition.

\begin{proposition} \label{proposition_1}
  \emph{
  Suppose $\{\boldsymbol{\varpi}^\star, \tilde{\mathbf{w}}^\star, \mathbf{r}_b^\star, \mathbf{r}_c^\star, \mathbf{p}^\star\}$ is a KKT optimal solution to problem (P1.2). If the optimal $\boldsymbol{\varpi}^\star$ and $\tilde{\mathbf{w}}^\star$ satisfy
  \begin{align}
    \label{condition:weight}
    &\boldsymbol{\varpi}^\star = \boldsymbol{\varpi}^{\text{MMSE}}(\mathbf{p}^\star),\\
    \label{condition:receiver}
    &\tilde{\mathbf{w}}^\star = \tilde{\mathbf{w}}^{\text{MMSE}}(\mathbf{p}^\star),
  \end{align}
  where $\boldsymbol{\varpi}^{\text{MMSE}} = [\varpi_{u,1}^{\text{MMSE}},\dots,\varpi_{u,k}^{\text{MMSE}}, \varpi_s^{\text{MMSE}}, \varpi_d^{\text{MMSE}}]^T$ and $\tilde{\mathbf{w}}^{\text{MMSE}} = [(\tilde{\mathbf{w}}_{u,1}^{\text{MMSE}})^T, \dots,(\tilde{\mathbf{w}}_{u,K}^{\text{MMSE}})^T,$ $(\tilde{\mathbf{w}}_s^{\text{MMSE}})^T, \tilde{w}_d^{\text{MMSE}}]^T$, then $\{\mathbf{r}_b^\star, \mathbf{r}_c^\star, \mathbf{p}^\star\}$ is a KKT optimal solution to problem (P1.1). 
  }
\end{proposition}

\begin{proof}
  Please refer to Appendix A.
\end{proof}

According to \textbf{Proposition \ref{proposition_1}}, a KKT optimal solution to problem (P1.1) can be obtained by finding a KKT optimal solution to problem (P1.2) that satisfies \eqref{condition:weight} and \eqref{condition:receiver}. It can be observed that although the optimization variables are still highly coupled in the constraints \eqref{constraint:up_rate_mmse}-\eqref{constraint:sensing_mmse}, (P1.2) is block-wise convex with respect to $\boldsymbol{\varpi}$, $\tilde{\mathbf{w}}$, and $\{\mathbf{r}_b, \mathbf{r}_c, \mathbf{p}\}$, which can be efficiently solved by invoking the alternating optimization (AO). The details of the AO algorithm are summarized in \textbf{Algorithm \ref{alg:A}}, and its optimality is given in the following proposition.

\begin{algorithm}[htb]
  \caption{WMMSE-based AO algorithm.}
  \label{alg:A}
  \begin{algorithmic}[1]
    \REQUIRE Convergence criteria and $\mathbf{p}^{[0]}$.
    \ENSURE $\mathbf{r}_b^\star$, $\mathbf{r}_c^\star$, $\mathbf{p}^\star$.
        \STATE{$n \coloneqq 0$.}
        \WHILE{\emph{$\tilde{R}$ not converged}}
          \STATE{update $\boldsymbol{\varpi}^{[n+1]}$ as $\boldsymbol{\varpi}^{\text{MMSE}}(\mathbf{p}^{[n]})$.}
          \STATE{update $\tilde{\mathbf{w}}^{[n+1]}$ as $\tilde{\mathbf{w}}^{\text{MMSE}}(\mathbf{p}^{[n]})$.}
          \STATE{update $\{\mathbf{r}_b^{[n+1]}, \mathbf{r}_c^{[n+1]}, \mathbf{p}^{[n+1]}\}$ by solving (P1.2) with $\boldsymbol{\varpi}^{[n+1]}$ and $\tilde{\mathbf{w}}^{[n+1]}$.}
          \STATE{$n \coloneqq n+1$.}
        \ENDWHILE
        \STATE{\textbf{return} $\mathbf{r}_b^\star = \mathbf{r}^{[n]}$, $\mathbf{r}_c^\star = \mathbf{r}^{[n]}$, and $\mathbf{p}^\star = \mathbf{p}^{[n]}$.}
    \end{algorithmic}
\end{algorithm}

\begin{proposition}\label{proposition_2}
  \emph{
    \textbf{Algorithm \ref{alg:A}} has a guaranteed convergence to the KKT optimal solutions to (P1.2) that satisfies \eqref{condition:weight} and \eqref{condition:receiver}.
  }
\end{proposition}

\begin{proof}
   Please refer to Appendix B.
\end{proof}

According to \textbf{Propositions \ref{proposition_1}} and \textbf{\ref{proposition_2}}, by solving (P1.2) through \textbf{Algorithm \ref{alg:A}}, the KKT optimal solutions to problem (P1.1) can be obtained. However, due to the sensing SINR constraint \eqref{constraint:sensing_mmse} in problem (P1.2), the initialized $\mathbf{p}^{[0]}$ may lead to infeasibility and thus to the failure of \textbf{Algorithm \ref{alg:A}}. To reduce the infeasibility probability as much as possible, we initialize $\mathbf{p}^{[0]}$ such that the main lobe of transmit beampattern towards the desired target direction $\theta_0$ and takes up all the available power, i.e.,
\begin{equation} \label{initial_p}
  \mathbf{p}^{[0]} = \frac{\sqrt{P_b} \mathbf{a}(\theta_0)}{\|\mathbf{a}(\theta_0)\|_2}.
\end{equation}
It should be pointed out that the optimal method to obtain $\mathbf{p}^{[0]}$ for avoiding infeasibility is to find the optimal $\mathbf{p}$ such that $\gamma_s$ is maximized. However, maximizing $\gamma_s$ with respect to $\mathbf{p}$ is a non-convex problem, which is hard to obtain the optimal solution. Nevertheless, the closed-form $\mathbf{p}^{[0]}$ given in \eqref{initial_p} can maximize the power in the target direction and achieve low power leakage in the other directions. In this case, the $\gamma_s$ achieved by the closed-form $\mathbf{p}^{[0]}$ will be very close to the maximum $\gamma_s$ in general. The complexity analysis of \textbf{Algorithm \ref{alg:A}} is given as follows. The complexities for updating $\boldsymbol{\varpi}$ and $\tilde{\mathbf{w}}$ are both $\mathcal{O}(K^2N^3)$, where $\mathcal{O}(\cdot)$ is the big-O notation. Then, updating $\{\mathbf{r}_b, \mathbf{r}_c, \mathbf{p}\}$ needs to solve the problem (P1.2) with the fixed $\boldsymbol{\varpi}$ and $\tilde{\mathbf{w}}$, which is a quadratically constrained quadratic program (QCQP) and thus has the complexity of $\mathcal{O}( K^{1/2} (3K+N)(2K+N)^2 )$ by applying the interior-point method \cite{nesterov1994interior}. Thus, the overall complexity of \textbf{Algorithm \ref{alg:A}} is $\mathcal{O}(I_{\mathrm{ite}}(2K^2N^3 + K^{1/2} (3K+N)(2K+N)^2))$, where $I_{\mathrm{ite}}$ denotes the number of iterations.

\subsection{Proposed Uplink NOMA Decoding Order of Computation Users} \label{sec:order}
In the above section, problem (P1.1) is solved with a given decoding order of NOMA users. A straightforward method to obtain the optimal decoding order is the exhaustive search 
which, however, has a prohibitive complexity of $\mathcal{O}(K!)$. As a remedy, we propose a user ordering scheme based on the weights of the computation rate and the large-scale fading. Note that the computation rate of user $k$ is upper bounded by the achievable rate $R_{u,k}$, which indicates that the user with higher weight requires higher $R_{u,k}$. To this end, the decoding order is first designed based on the weights of users as follows:
\begin{equation}
  \pi(i) < \pi(j), \quad \text{if } \omega_i < \omega_j, \quad \forall i,j \in \mathcal{K}.
\end{equation}
Based on this decoding order, the users with the higher weight suffer less interference from other users, leading to the higher $R_{u,k}$. Next, for users with the same weights, the decoding order is further determined by the path loss, which is given by 
\begin{equation}
  \pi(i) < \pi(j), \quad \text{if } \omega_i = \omega_j \text{ and } L_{u,i} \ge L_{u,j}, \quad \forall i,j \in \mathcal{K}.
\end{equation}
The effectiveness of the proposed user ordering scheme will be verified by the numerical results in Section \ref{sec:results}.
\section{Binary Computation Offloading} \label{sec:binary_offloading}
In this section, the binary offloading mode is considered, where the binary offloading decision is jointly optimized with the computing resource allocation $\{r_{b,k}\}$ and $\{r_{c,k}\}$ and the transmit beamformer $\mathbf{p}$ for maximizing the computation rate $\tilde{R}$. We propose an ADMM-based AO algorithm to address the coupled optimization variables and the binary integer constraint.

\subsection{Problem Formulation}
For the binary computation offloading, where $\mathcal{X}_{\mathrm{BS}} \subseteq \mathcal{K}$ and $\mathcal{X}_{\mathrm{CS}} = \mathcal{K} \backslash \mathcal{X}_{\mathrm{BS}}$, the computation rate maximization problem with a given decoding order $\pi$ can be formulated as follows:
\begin{subequations}
    \begin{align}
        \text{(P2):} \quad \max_{\mathcal{X}_{\mathrm{BS}}, \mathbf{r}_b, \mathbf{r}_c, \mathbf{p}} \quad &  \sum_{k \in \mathcal{X}_{\mathrm{BS}}} \omega_k r_{b,k} + \sum_{k \in \mathcal{X}_{\mathrm{CS}}} \omega_k r_{c,k} \\
        \label{constraint:set}
        \mathrm{s.t.} \quad & \mathcal{X}_{\mathrm{BS}} \subseteq \mathcal{K}, \mathcal{X}_{\mathrm{CS}} = \mathcal{K} \backslash \mathcal{X}_{\mathrm{BS}}, \\
        & \eqref{constraint:causality_1} - \eqref{constraint:positive}.
    \end{align}
\end{subequations}
It can be observed that for the fixed $\mathcal{X}_{\mathrm{BS}}$, problem (P2) is equivalent to problem (P1) with additional constraints $r_{b,k} = 0, \forall k \in \mathcal{X}_{\mathrm{CS}}$ and $r_{c,k} = 0, \forall k \in \mathcal{X}_{\mathrm{BS}}$, which can be solved by the proposed \textbf{Algorithm \ref{alg:A}}. Then, the optimal $\mathcal{X}_{\mathrm{BS}}$ can be obtained by exhaustively searching all the subsets of $\mathcal{K}$. However, the exhaustive search requires a prohibitive complexity of $\mathcal{O}(2^{K})$. As a remedy, we propose an ADMM-based AO algorithm with much lower complexity to solve the problem (P2) in the following subsections.

\subsection{Reformulation of Problem (P2)}
Before introducing the ADMM-based AO algorithm, we first reformulate problem (P2) to an equivalent but more mathematically tractable form. To this end, we introduce a set of binary decision variables $\mathbf{m} = [m_1,\dots,m_K]^T$ and reformulate problem (P2) as follows:
\begin{subequations}
    \begin{align}
        \text{(P2.1):} \quad &\max_{\mathbf{m}, \mathbf{r}_b, \mathbf{r}_c, \mathbf{p}} \quad \sum_{k \in \mathcal{K}} \omega_k \big( m_k r_{b,k} + (1-m_k) r_{c,k} \big) \\
        \label{constraint:coupled_1}
        \mathrm{s.t.} \quad & m_k r_{b,k} + (1-m_k)r_{c,k} \le B R_{u, k}, \forall k \in \mathcal{K}, \\
        \label{constraint:coupled_2}
        & \sum_{k \in \mathcal{K}} (1-m_k) r_{c,k}\le B R_d, \\
        \label{constraint:coupled_3}
        & \sum_{k \in \mathcal{K}} \kappa (\phi m_k r_{b,k})^3+ \mathrm{tr}(\mathbf{p}\mathbf{p}^H) \le P_b,\\
        \label{constraint:ingeter}
        & m_k \in \{0, 1\}, \forall k \in \mathcal{K},\\
        & \eqref{constraint:sensing_SINR},\eqref{constraint:positive},
    \end{align}
\end{subequations}
where $m_k=1$ for $k \in \mathcal{X}_{\mathrm{BS}}$ and $m_k=0$ for $k \in \mathcal{X}_{\mathrm{CS}}$. Problem (P2.1) is a non-convex mixed-integer optimization problem, the main challenge for which lies in the binary constraint \eqref{constraint:ingeter} and the coupling between $\{\mathbf{r}_b, \mathbf{r}_c\}$ and $\mathbf{m}$. To handle these obstacles, we introduce another set of variables $\mathbf{z}_b = [z_{b,1},\dots,z_{b,K}]^T$ and $\mathbf{z}_c = [z_{c,1},\dots,z_{c,K}]^T$ such that $z_{b,k} = m_k r_{b,k}$ and $z_{c,k} = (1-m_k) r_{c,k}$. Next, we transform the binary constraint \eqref{constraint:ingeter} into a series of equivalent equality constraints as follows:
\begin{equation}
  m_k = \tilde{m}_k, \quad m_k(1 - \tilde{m}_k) = 0, \quad 0 \le m_k \le 1, \quad \forall k \in \mathcal{K}.
\end{equation} 
It can be easily verified that the $m_k$ satisfying the above constraint must be either $0$ or $1$. Then, following the same path as Section \ref{sec:partial_wmmse} and applying the rate-WMMSE relationships, problem (P2.1) can be equivalently reformulated as follows:
\begin{subequations}
  \begin{align}
      \text{(P2.2):} \quad &\max_{\boldsymbol{\theta}} \quad  
      f(\mathbf{z}_b, \mathbf{z}_c) = \sum_{k \in \mathcal{K}} \omega_k \big( z_{b,k} + z_{c,k} \big) \\
      \label{constraint:2.2_start}
      \mathrm{s.t.} \quad & z_{b,k} + z_{c,k} \le B c - B \varepsilon_{u,k}, \forall k \in \mathcal{K}, \\
      & \sum_{k \in \mathcal{K}} z_{c,k} \le B c - B \varepsilon_d, \\
      & \Gamma_{r,\min} \le c - \varepsilon_r,\\
      & \sum_{k \in \mathcal{K}} \kappa (\phi z_{b,k})^3 + \mathrm{tr}(\mathbf{p}\mathbf{p}^H) \le P_b,\\
      & z_{b,k} \ge 0, z_{c,k} \ge 0, \forall k \in \mathcal{K}, \\
      \label{constraint:2.2_end}
      & 0 \le m_k \le 1, \forall k \in \mathcal{K},\\
      \label{constraint:2.3.5}
      & z_{b,k} = m_k r_{b,k}, z_{c,k} = (1-m_k) r_{c,k}, \forall k \in \mathcal{K},\\
      \label{constraint:2.3.6}
      &m_k = \tilde{m}_k, m_k(1-\tilde{m}_k) = 0, \forall k \in \mathcal{K},
  \end{align}
\end{subequations}
where $\boldsymbol{\theta} = \{\boldsymbol{\varpi}, \tilde{\mathbf{w}}, \mathbf{r}_b, \mathbf{r}_c, \tilde{\mathbf{m}}, \mathbf{z}_b, \mathbf{z}_c, \mathbf{p}, \mathbf{m}\}$ denotes all the optimization variables of (P2.2). In essence, (P2.2) is equivalent to (P1.2) with the additional highly coupled equality constraints. Thus, in the following subsection, we propose to solve it by invoking the ADMM optimization framework.

\subsection{ADMM-based AO Algorithm for Solving (P2.2)}
To accommodate the ADMM optimization framework, we equivalently transform (P2.2) into the following form:
\begin{subequations}
  \begin{align}
      \text{(P2.3):} \quad \min_{\boldsymbol{\theta}} \quad & 
      -f(\mathbf{z}_b, \mathbf{z}_c) + g(\boldsymbol{\theta}) \\
      \label{constraint:2.3.1}
      \mathrm{s.t.} \quad & \mathbf{z}_b = \mathbf{M} \mathbf{r}_b, \quad \mathbf{z}_c = (\mathbf{I} - \mathbf{M}) \mathbf{r}_c,\\
      \label{constraint:2.3.2}
      & \mathbf{m} = \tilde{\mathbf{m}}, \quad \mathbf{M}(\mathbf{1}-\tilde{\mathbf{m}}) = \mathbf{0},
  \end{align}
\end{subequations}
where $\mathbf{M} = \mathrm{diag}(\mathbf{m})$ and $g(\boldsymbol{\theta})$ is an indicator function for the constraints \eqref{constraint:2.2_start}-\eqref{constraint:2.2_end}. Specifically, if $\boldsymbol{\theta}$ falls in the feasible region defined by \eqref{constraint:2.2_start}-\eqref{constraint:2.2_end}, then $g(\boldsymbol{\theta}) = 0$, otherwise $g(\boldsymbol{\theta}) = +\infty$. Constraints \eqref{constraint:2.3.1} and \eqref{constraint:2.3.2} are transformed from \eqref{constraint:2.3.5} and \eqref{constraint:2.3.6}, respectively. Although problem (P2.3) is non-convex, the effectiveness of the ADMM framework to obtain the locally optimal solutions to the non-convex problems has been proved in \cite{hong2016convergence}. The ADMM update steps for solving problem (P2.3) can be formulated based on the following augmented Lagrangian:
\begin{align}
  \mathcal{L}(\boldsymbol{\theta}, \boldsymbol{\lambda}) = &-f(\mathbf{z}_b, \mathbf{z}_c) + g(\boldsymbol{\theta}) + \Pi(\boldsymbol{\theta})
\end{align}
with 
\begin{align}
  &\Pi(\boldsymbol{\theta}, \boldsymbol{\lambda}) \nonumber \\
  = &\frac{1}{2 \rho_1}\|\mathbf{z}_b - \mathbf{M} \mathbf{r}_b + \rho_1 \boldsymbol{\lambda}_1\|_2^2 + \frac{1}{2 \rho_2}\| \mathbf{z}_c - (\mathbf{I} - \mathbf{M})\mathbf{r}_c + \rho_2 \boldsymbol{\lambda}_2 \|_2^2 \nonumber \\
  &+ \frac{1}{2 \rho_3}\|\mathbf{m} - \tilde{\mathbf{m}} + \rho_3 \boldsymbol{\lambda}_3 \|_2^2 + \frac{1}{2 \rho_4}\|\mathbf{M}(\mathbf{1} - \tilde{\mathbf{m}}) + \rho_4 \boldsymbol{\lambda}_4\|_2^2,
\end{align}
where $\boldsymbol{\lambda} = \{\boldsymbol{\lambda}_1,\boldsymbol{\lambda}_2,\boldsymbol{\lambda}_3,\boldsymbol{\lambda}_4\}$ denote the Lagrangian dual variables associated with the constraints \eqref{constraint:2.3.1} and \eqref{constraint:2.3.2}, and $\{\rho_1, \rho_2, \rho_3, \rho_4\}$ denote the penalty parameters. Following the ADMM process, the primal variables $\boldsymbol{\theta}$ can be updated to minimize $\mathcal{L}(\boldsymbol{\theta}, \boldsymbol{\lambda})$ by invoking the block coordinate descent (BCD), while the dual variables $\boldsymbol{\lambda}$ can be updated to maximize $\mathcal{L}(\boldsymbol{\theta}, \boldsymbol{\lambda})$ based on the gradient descent (GD). More specifically, the primal variables $\boldsymbol{\theta}$ are divided into three blocks, namely $\{\boldsymbol{\varpi}, \tilde{\mathbf{w}}\}$, $\{\mathbf{r}_b, \mathbf{r}_c, \tilde{\mathbf{m}}\}$, and $\{ \mathbf{z}_b, \mathbf{z}_c, \mathbf{p}, \mathbf{m}\}$, for BCD. Therefore, at the $n$-th iteration, the following ADMM update steps are formulated.

\textbf{Update $\{\boldsymbol{\varpi}^{[n+1]},\tilde{\mathbf{w}}^{[n+1]}\}$}:
Following the same path in Section \ref{sec:partial_wmmse}, the block $\{\boldsymbol{\varpi}, \tilde{\mathbf{w}}\}$ can be updated by
\begin{subequations} \label{admm_step_1}
  \begin{align}
    &\boldsymbol{\varpi}^{[n+1]} = \boldsymbol{\varpi}^{\text{MMSE}}(\mathbf{p}^{[n]}),\\
    &\tilde{\mathbf{w}}^{[n+1]} = \tilde{\mathbf{w}}^{\text{MMSE}}(\mathbf{p}^{[n]}).
  \end{align}
\end{subequations}

\textbf{Update $\{ \mathbf{r}_b^{[n+1]}, \mathbf{r}_c^{[n+1]}, \tilde{\mathbf{m}}^{[n+1]}\}$}: The subproblem with respect to $\{\mathbf{r}_b, \mathbf{r}_c, \tilde{\mathbf{m}}\}$ is given by
\begin{equation}
  \{\mathbf{r}_b^\star, \mathbf{r}_c^\star, \tilde{\mathbf{m}}^\star\} = \operatorname*{argmin}_{\mathbf{r}_b, \mathbf{r}_c, \tilde{\mathbf{m}}} \Pi(\boldsymbol{\theta}, \boldsymbol{\lambda})
\end{equation}
which is an unconstrained optimization problem. By checking the first-order optimality conditions, the update of $\{ \mathbf{r}_b^{[n+1]}, \mathbf{r}_c^{[n+1]}, \tilde{\mathbf{m}}^{[n+1]}\}$ can be obtained, which is given by
\begin{subequations}  \label{admm_step_2}
  \begin{align}
    \label{eqn:optimal_r}
    r_{b,k}^{[n+1]} = &\frac{z_{b,k}^{[n]} + \rho_1 \lambda_{1,k}^{[n]}}{m_k^{[n]}}, r_{c,k}^{[n+1]} = \frac{z_{c,k}^{[n]} + \rho_2 \lambda_{2,k}^{[n]}}{1 - m_k^{[n]}}, \forall k \in \mathcal{K}, \\
    \tilde{\mathbf{m}}^{[n+1]} = &\big( (\mathbf{M}^{[n]})^T \mathbf{M}^{[n]} + \mathbf{I} \big)^{-1} \big( (\mathbf{M}^{[n]})^T \mathbf{m}^{[n]} \nonumber \\
    &+ \rho_4 (\mathbf{M}^{[n]})^T \boldsymbol{\lambda}_4^{[n]} +  \mathbf{m}^{[n]} + \rho_3 \boldsymbol{\lambda}_3^{[n]} \big),
  \end{align}
\end{subequations}
where $\lambda_{i,k}^{[n]}$ denotes the $k$-th entry of $\boldsymbol{\lambda}_i^{[n]}$. Note that the update of $r_{b,k}^{[n+1]}$ and $r_{c,k}^{[n+1]}$ in \eqref{eqn:optimal_r} only hold when $m_k^{[n+1]} \neq 0$ and $m_k^{[n+1]} \neq 1$, respectively. However, when $m_k^{[n+1]}=0$, $r_{b,k}^{[n+1]}$ can be set to an arbitrary value since $r_{b,k}^{[n+1]}$ becomes a dummy variable and has no influence on the objective value. Similarly, when $m_k^{[n+1]} = 1$, $r_{c,k}^{[n+1]}$ can also be set to an arbitrary value.    

\textbf{Update $\{ \mathbf{z}_b^{[n+1]}, \mathbf{z}_c^{[n+1]}, \mathbf{p}^{n+1}, \mathbf{m}^{[n+1]}\}$}:
The optimization problem with respect to $\{ \mathbf{z}_b, \mathbf{z}_c, \mathbf{p}, \mathbf{m}\}$ with the fixed remaining blocks is given by
\begin{equation} \label{admm_step_3}
  \{ \mathbf{z}_b^\star, \mathbf{z}_c^\star, \mathbf{p}^\star, \mathbf{m}^\star\} = \operatorname*{argmin}_{\{ \mathbf{z}_b, \mathbf{z}_c, \mathbf{p}, \mathbf{m}\} \in \mathcal{S}} -f(\mathbf{z}_b, \mathbf{z}_c) + \Pi(\boldsymbol{\theta}, \boldsymbol{\lambda}),
\end{equation}
where $\mathcal{S}$ denotes the feasible region defined by \eqref{constraint:2.2_start}-\eqref{constraint:2.2_end}. The above problem is a convex QCQP and can be optimally solved by the existing toolbox, such as CVX \cite{cvx}.

\textbf{Update  $\boldsymbol{\lambda}^{[n+1]}$}: After updating all the blocks in $\boldsymbol{\theta}$, the GD update of the dual variables is given by 
\begin{subequations} \label{update_dual}
  \begin{align}
    &\boldsymbol{\lambda}_1^{[n+1]} = \boldsymbol{\lambda}_1^{[n]} + \frac{1}{\rho_1} \left(\mathbf{z}_b^{[n+1]} - \mathbf{M}^{[n+1]} \mathbf{r}_b^{[n+1]}\right),\\
    &\boldsymbol{\lambda}_2^{[n+1]} = \boldsymbol{\lambda}_2^{[n]} + \frac{1}{\rho_2} \left( \mathbf{z}_c^{[n+1]} - (\mathbf{I} - \mathbf{M}^{[n+1]})\mathbf{r}_c^{[n+1]}\right),\\
    &\boldsymbol{\lambda}_3^{[n+1]} = \boldsymbol{\lambda}_3^{[n]} + \frac{1}{\rho_3} \left(\mathbf{m}^{[n+1]} - \tilde{\mathbf{m}}^{[n+1]}\right),\\
    &\boldsymbol{\lambda}_4^{[n+1]} = \boldsymbol{\lambda}_4^{[n]} + \frac{1}{\rho_4} \mathbf{M}^{[n+1]} (\mathbf{1} - \tilde{\mathbf{m}}^{[n+1]}).
  \end{align}
\end{subequations}
The details of the proposed ADMM-based AO algorithm for solving (P2.3) are summarized in \textbf{Algorithm \ref{alg:B}}, where $\delta$ is the constraint violation function defined as follows:
\begin{equation}
  \delta = \max \begin{Bmatrix}
    \|\mathbf{z}_b - \mathbf{M} \mathbf{r}_b\|_\infty, \| \mathbf{z}_c - (\mathbf{I} - \mathbf{M})\mathbf{r}_c\|_\infty, \\ 
    \|\mathbf{m} - \tilde{\mathbf{m}}\|_\infty, \|\mathbf{M} (\mathbf{1} - \tilde{\mathbf{m}})\|_\infty
  \end{Bmatrix}
\end{equation}

\begin{algorithm}[htb]
  \caption{ADMM-based AO algorithm.}
  \label{alg:B}
  \begin{algorithmic}[1]
    \REQUIRE ADMM parameters and $\{ \mathbf{z}_b^{[0]}, \mathbf{z}_c^{[0]}, \mathbf{p}^{[0]}, \mathbf{m}^{[0]} \}$.
    \ENSURE $\mathbf{z}_b^\star$, $\mathbf{z}_c^\star$, $\mathbf{p}^\star$, $\mathbf{m}^{\star}$.
        \STATE{$n \coloneqq 0$.}
        \WHILE{\emph{$\tilde{R}$ not converged} or $\delta > \tau$}
          \STATE{update $\{\boldsymbol{\varpi}^{[n+1]}, \tilde{\mathbf{w}}^{[n+1]}\}$ as shown in \eqref{admm_step_1}}
          \STATE{update $\{\mathbf{r}_b^{[n+1]}, \mathbf{r}_c^{[n+1]}, \tilde{\mathbf{m}}^{[n+1]}\}$ as shown in \eqref{admm_step_2}.}
          \STATE{update $\{ \mathbf{z}_b^{[n+1]}, \mathbf{z}_c^{[n+1]}, \mathbf{p}^{[n+1]}, \mathbf{m}^{[n+1]}\}$ by solving \eqref{admm_step_3} with the fixed remaining blocks obtained in step 3 and step 4.}
          \STATE{update dual variables $\boldsymbol{\lambda}^{[n+1]}$ as shown in \eqref{update_dual}.}
          \STATE{$n \coloneqq n+1$.}
        \ENDWHILE
        \STATE{\textbf{return} $\mathbf{z}_b^\star = \mathbf{z}_b^{[n]}$, $\mathbf{z}_c^\star = \mathbf{z}_c^{[n]}$, and $\mathbf{p}^\star = \mathbf{p}^{[n]}$, $\mathbf{m}^{\star} = \mathbf{m}^{[n]}$.}
    \end{algorithmic}
\end{algorithm}

To accelerate the convergence and enhance the performance of \textbf{Algorithm \ref{alg:B}}, the initialization point needs to be appropriately selected. By observing that problem (P2.2) equals to problem (P1.2) with the additional equality constraints, we initialize the optimization variables of \textbf{Algorithm \ref{alg:B}} according to the output of \textbf{Algorithm \ref{alg:A}}. Let $\breve{\mathbf{r}}_b$, $\breve{\mathbf{r}}_c$, and $\breve{\mathbf{p}}$ denote the output of \textbf{Algorithm \ref{alg:A}}. Then, the $\mathbf{z}_b^{[0]}$, $\mathbf{z}_c^{[0]}$, $\mathbf{p}^{[0]}$, and $\mathbf{m}^{[0]}$ of \textbf{Algorithm \ref{alg:B}} can be initialized as follows:
\begin{equation} \label{initial_z_p_m}
  \mathbf{z}_b^{[0]} = \breve{\mathbf{r}}_b, \mathbf{z}_c^{[0]} = \breve{\mathbf{r}}_c, \mathbf{p}^{[0]} = \breve{\mathbf{p}}, \mathbf{m}^{[0]} = \breve{\mathbf{r}}_b \oslash (\breve{\mathbf{r}}_b + \breve{\mathbf{r}}_c),
\end{equation}
where the operation $\oslash$ denotes the element-wise division. Moreover, the decoding order proposed in Section \ref{sec:order} can also be applied for binary offloading.
Note that in \textbf{Algorithm \ref{alg:B}}, each block of $\{\boldsymbol{\varpi}, \tilde{\mathbf{w}}\}$, $\{\mathbf{r}_b, \mathbf{r}_c, \tilde{\mathbf{m}}\}$, and $\{ \mathbf{z}_b, \mathbf{z}_c, \mathbf{p}, \mathbf{m}\}$ is updated as the optimal solution with the fixed other blocks, indicating that the objective value of the augmented Lagrangian problem is non-increasing over the iterations. Therefore, the convergence of \textbf{Algorithm \ref{alg:B}} is ensured. Although the optimality of the obtained solution via \textbf{Algorithm \ref{alg:B}} is difficult to mathematically prove, our numerical results in the following section confirm that the obtained binary-offloading solution is close to the optimal one obtained by exhaustive search. The complexity analysis of \textbf{Algorithm \ref{alg:B}} is given as follows. The computational complexities for updating $\{\boldsymbol{\varpi}, \tilde{\mathbf{w}}\}$ and $\{\mathbf{r}_b, \mathbf{r}_c, \tilde{\mathbf{m}}\}$ are $\mathcal{O}(2K^2 N^3)$ and $\mathcal{O}(2K^3)$, respectively. By applying the interior-point method \cite{nesterov1994interior}, the worst case computational complexities for updating $\{\mathbf{z}_b, \mathbf{z}_c, \mathbf{p}, \mathbf{m}\}$ is $\mathcal{O}( K^{1/2} (4K+N)(3K+N)^2 )$. Finally, the complexity for updating the dual variables $\boldsymbol{\lambda}$ and $\boldsymbol{\mu}$ is $\mathcal{O}(3K^2)$. Therefore, the overall complexity of \textbf{Algorithm \ref{alg:B}} is $\mathcal{O}(I_{\mathrm{ite}}(2K^2 N^3 + 2K^3 + K^{1/2} (4K+N)(3K+N)^2 + 6K^2))$.
\section{Numerical Results} \label{sec:results}

In this section, we provide the numerical results obtained by Monte Carlo simulations for characterizing the proposed NOMA-aided JCSMC framework. In particular, the BS is assumed to be equipped with a ULA with $N=8$ antennas and half wavelength spacing. We assume that there are $K=3$ computation users, one desired target under the angle of $0^\circ$, and $L=4$ clutter sources under the angles of $\{-60^{\circ}, -30^{\circ}, 30^{\circ}, 60^{\circ}\}$. The channels between the BS and computation users are assumed to obey Rayleigh fading with the path loss of $L_{u,k}(\text{dB}) = 30 + 30 \log_{10} (d_{u,k})$, where $d_{C,k}$ denotes the distance between the BS and computation user $k$. The channels between the BS and target/clutters are assumed to have pure LoS components with the similarly path loss of $L_{s,m}(\text{dB}) = 30 + 30 \log_{10} (2d_{s,m})$ and the same complex reflection factor of $|\tilde{\beta}_m|=1$, where $d_{s,m}$ denotes the distance between the BS and $m$-th target/clutters. The parameters in simulations are set as follows: $d_{u,k} = 60 \text{ m}, \forall k \in \mathcal{K}$, and $d_{s,m} = 50 \text{ m}, \forall m \in \mathcal{M}$. The noise power at the receivers of the BS and the CS is assumed to be $\sigma_u^2=\sigma_d^2 = -80$ dBm. The communication bandwidth is assumed to be $B=30$ MHz. For the computation task, we set $\phi = 3 \times 10^3$ cycle/bit and $\kappa = 10^{-26}$ \cite{wang2018multi}. The penalty factors of \textbf{Algorithm \ref{alg:B}} are set as $\rho_1 = \rho_2 = 1$ and $\rho_3 = \rho_4 = 0.1$. The above parameter setup is assumed throughout our simulations unless otherwise specified. 

To prove the effectiveness of the proposed NOMA-aided JCSMC framework, the following benchmark schemes are considered for comparison.
\begin{itemize}
    \item \textbf{BS-only/CS-only computing}:
    This scheme is based on the single-tier computing structure, i.e., all the computation tasks from users are only executed at BS or CS. The corresponding computation rate can be obtained by solving problem (P2.1) with $m_k=1, \forall k \in \mathcal{K}$ (at BS) or $m_k=0, \forall k \in \mathcal{K}$ (at CS) through \textbf{Algorithm \ref{alg:A}}.

    \item \textbf{Space-division multiple access (SDMA)-aided framework}: 
    In this scheme, the BS decodes the offloading signals from each user subject to interference from all the remaining users with MMSE receivers. Thus, the covariance matrix of the effective noise is given by
    \begin{align} \label{sdma_1}
        \mathbf{R}_{\tilde{\mathbf{n}}_k}^{\text{SDMA}} &= \sum_{i \in \mathcal{K}, i \neq k} P_u \mathbf{h}_{u,i} \mathbf{h}_{u,i}^H + (\mathbf{H}_s+\mathbf{H}_c) \mathbf{p} \mathbf{p}^H  \nonumber \\
        &\times (\mathbf{H}_s + \mathbf{H}_c)^H + \sigma_{u}^2 \mathbf{I}_N.
    \end{align}
    Accordingly, the achievable computation offloading rate from user $k$ to the BS with the MMSE receiver is $R_{u,k}^{\text{SDMA}} = \log_2(1 + P_u \mathbf{h}_{u,k}^H (\mathbf{R}_{\tilde{\mathbf{n}}_k}^{\text{SDMA}})^{-1} \mathbf{h}_{u,k} )$. Moreover, as the SIC is not exploited in SDMA, the interference from the computation offloading signal to the echo sensing signal cannot be removed. Thus, the covariance matrix of the effective noise for the echo sensing signal is given by
    \begin{equation} \label{sdma_2}
        \mathbf{R}_{\mathbf{c}}^{\text{SDMA}} = \sum_{k \in \mathcal{K}} P_u \mathbf{h}_{u,k} \mathbf{h}_{u,k}^H + \mathbf{H}_c \mathbf{p} \mathbf{p}^H \mathbf{H}_c^H + \sigma_{u}^2 \mathbf{I}_N.
    \end{equation}
    Then, the corresponding sensing SINR is $\gamma_s^{\text{SDMA}} = \mathbf{p}^H \mathbf{H}_s^H (\mathbf{R}_{\mathbf{c}}^{\text{SDMA}})^{-1} \mathbf{H}_s \mathbf{p}$.
    The related optimization problems for partial and binary computation offloading modes can be solved by using \textbf{Algorithm \ref{alg:A}} and \textbf{Algorithm \ref{alg:B}}, respectively.
\end{itemize}

\begin{figure}[t!]
    \centering
    \includegraphics[width=0.4\textwidth]{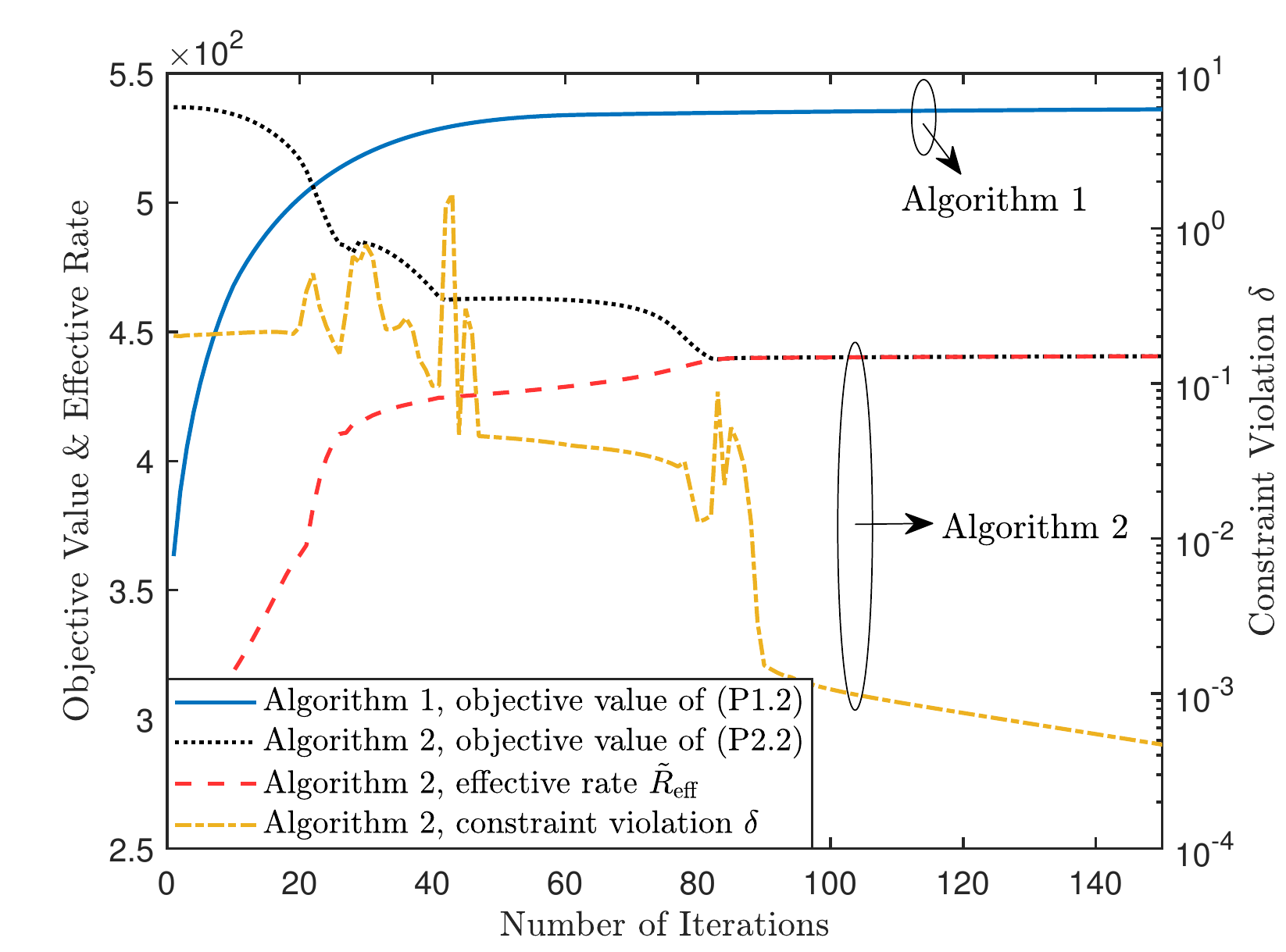}
    \caption{Convergence behavior of \textbf{Algorithms \ref{alg:A}} and \textbf{\ref{alg:B}}.}
    \label{fig:convergence_WMMSE}
    \vspace{-0.4cm}
\end{figure}

In Fig. \ref{fig:convergence_WMMSE}, we first demonstrate the convergence behavior of the proposed \textbf{Algorithms \ref{alg:A}} and \textbf{\ref{alg:B}}. It can be observed that the computation rate obtained by \textbf{Algorithm \ref{alg:A}} monotonically increases in each iteration and finally converges to a stationary value. Next, at the $n$-th iteration of \textbf{Algorithm \ref{alg:B}}, we define the effective computation rate with the binary offloading mode as follows:
\begin{equation}
    \tilde{R}_{\mathrm{eff}}^{[n]} = \sum_{k \in \mathcal{K}} \omega_k \left( \breve{m}_k^{[n]} z_{b,k}^{[n]} + (1-\breve{m}_k^{[n]}) z_{c,k}^{[n]}  \right),
\end{equation}
where $\breve{m}_k^{[n]} = 0$ if $m_k^{[n]} \le 0.5$, otherwise $\breve{m}_k^{[n]} = 1$. It can be seen that for the \textbf{Algorithm \ref{alg:B}} both the objective value of (P2.2), i.e., $\sum_{k \in \mathcal{K}} \omega_k \big( z_{b,k}^{[n]} + z_{c,k}^{[n]} \big)$, and the effective rate converge to a stationary value. Meanwhile, the gap between them eventually becomes negligible since the constraint violation $\delta$ approaches zero as the iterations progress. 

In Fig. \ref{fig:effectiveness_admm}, we examine the quality of the obtained binary offloading scheme via \textbf{Algorithm \ref{alg:B}} and the proposed NOMA user ordering scheme, by comparing them with the optimal exhaustive search scheme and the random selection scheme. As can be observed, the obtained binary offloading scheme achieves a very close performance to the optimal exhaustive search scheme and significantly outperforms the random selection scheme. Then, to evaluate the effectiveness of the proposed NOMA user ordering scheme, we consider two scenarios. One is with different weights of the computation rate and equal path loss of the computation users, while the other is the opposite. It is illustrated that the proposed NOMA user ordering scheme has the same performance as the optimal exhaustive search scheme in both scenarios. Furthermore, it is interesting to observe that the random selection scheme suffers performance loss if only the weights are different, but it does not if only the path loss is different, which implies that the path loss makes no difference in the decoding order.

\begin{figure}[t!]
    \centering
    \includegraphics[width=0.4\textwidth]{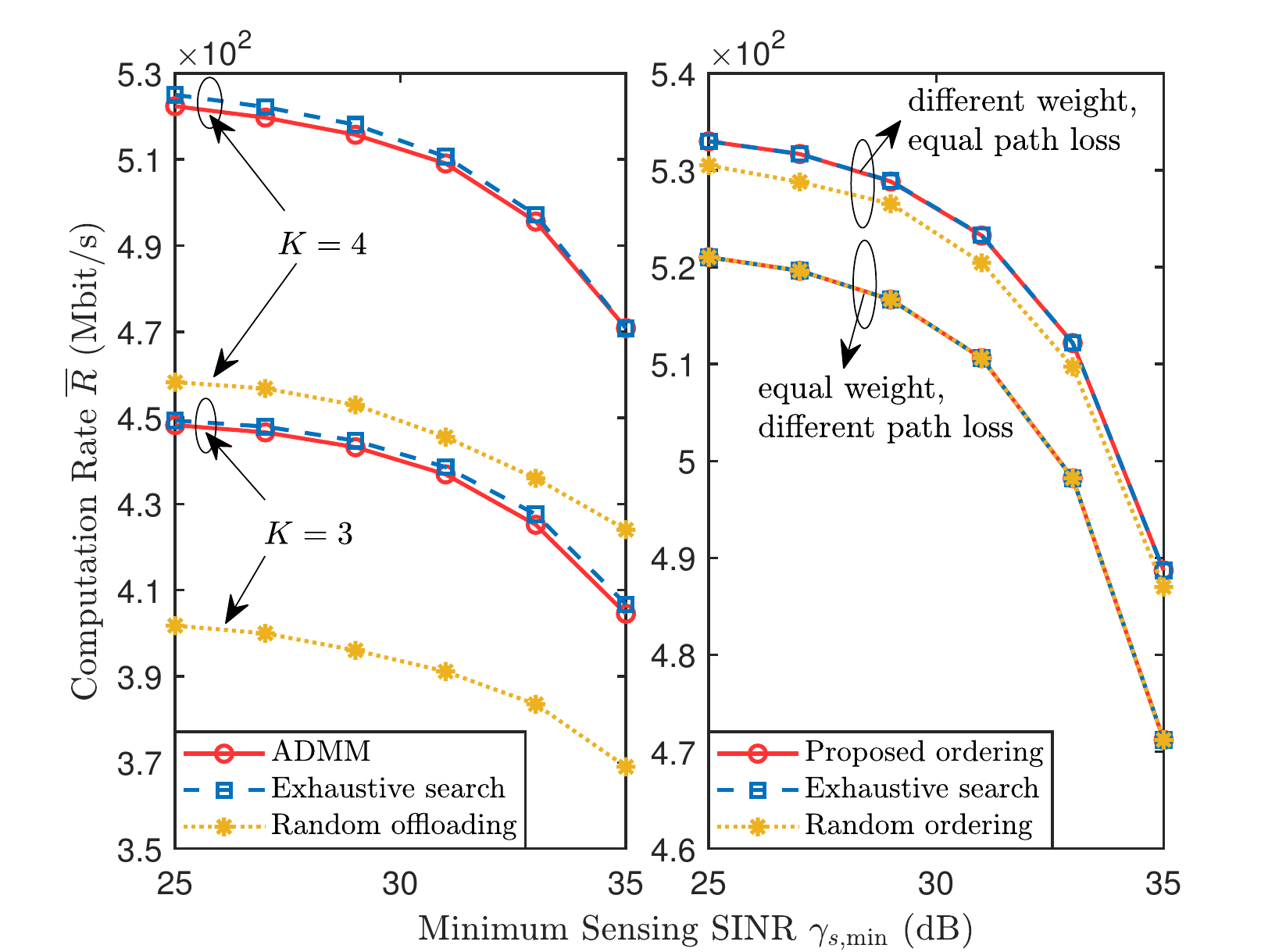}
    \caption{Efficiency of the proposed binary offloading scheme (left) and NOMA user ordering scheme (right).}
    \label{fig:effectiveness_admm}
    \vspace{-0.4cm}
\end{figure}
\begin{figure}[t!]
    \centering
    \includegraphics[width=0.4\textwidth]{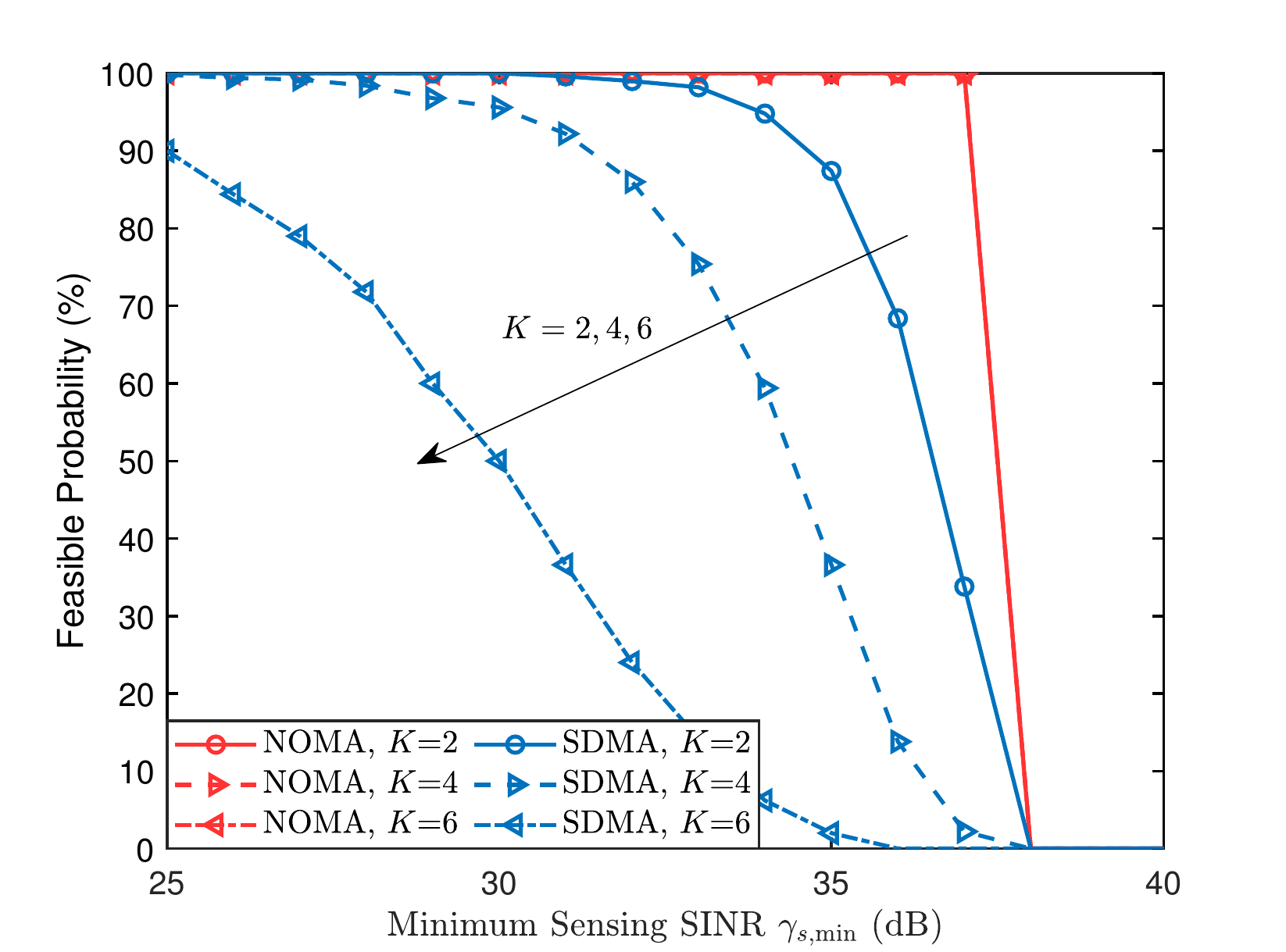}
    \caption{Feasible probability of (P1.1) and (P2.1) versus sensing SINR $\gamma_{s,\min}$.}
    \label{fig:feasible}
    \vspace{-0.4cm}
\end{figure}

It is worth noting that the problems (P1.1) and (P2.1) are not always feasible for both NOMA and SDMA because of the minimum sensing SINR constraint. 
Consequently, in Fig. \ref{fig:feasible}, we illustrate the successful probability of obtaining a feasible solution to problems (P1.1) and (P2.1) via the proposed algorithms versus the sensing SINR. The results are obtained over $1000$ random channel realizations. We can observe that when NOMA is employed, a feasible solution to problems (P1.1) and (P2.1) can always be obtained if $\gamma_{s,\min} \le 37$ dB. However, when SDMA is employed, the successful probability decreases as $K$ and $\gamma_{s,\min}$ increase. This is because the SDMA-assisted framework does not eliminate the interference of the offloading signals to the sensing signals. These results highlight the necessity of NOMA in terms of guaranteeing the feasibility of the JCSMC framework in practice. In the following simulations, to make a fair comparison between NOMA and SDMA, we only keep the feasible results.

\begin{figure}[t!]
    \centering
    \includegraphics[width=0.4\textwidth]{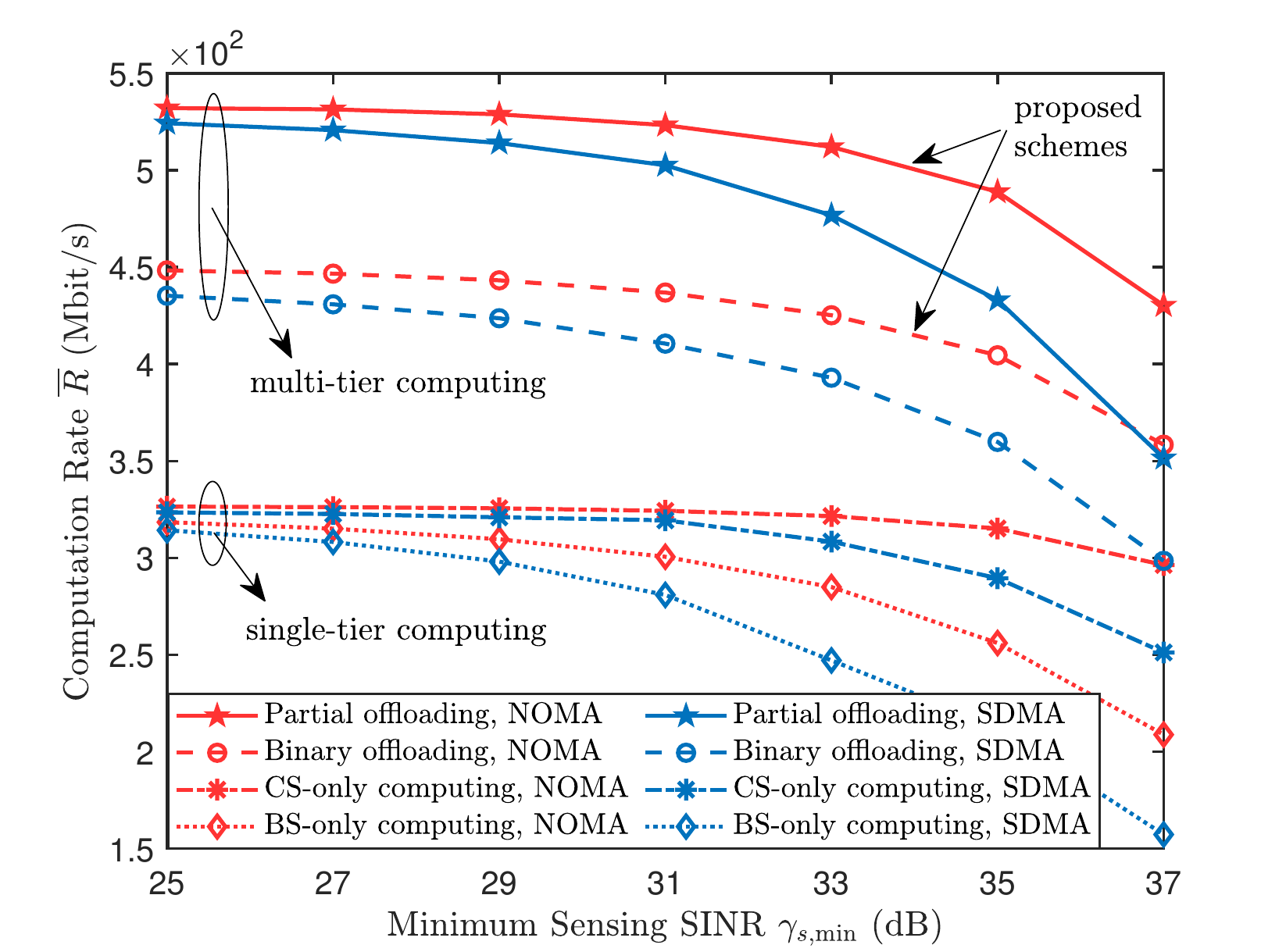}
    \caption{Computation rate $\tilde{R}$ versus sensing SINR $\gamma_{s, \min}$.}
    \label{fig:rate_vs_SINR}
    \vspace{-0.4cm}
\end{figure}
\begin{figure}[t!]
    \centering
    \includegraphics[width=0.4\textwidth]{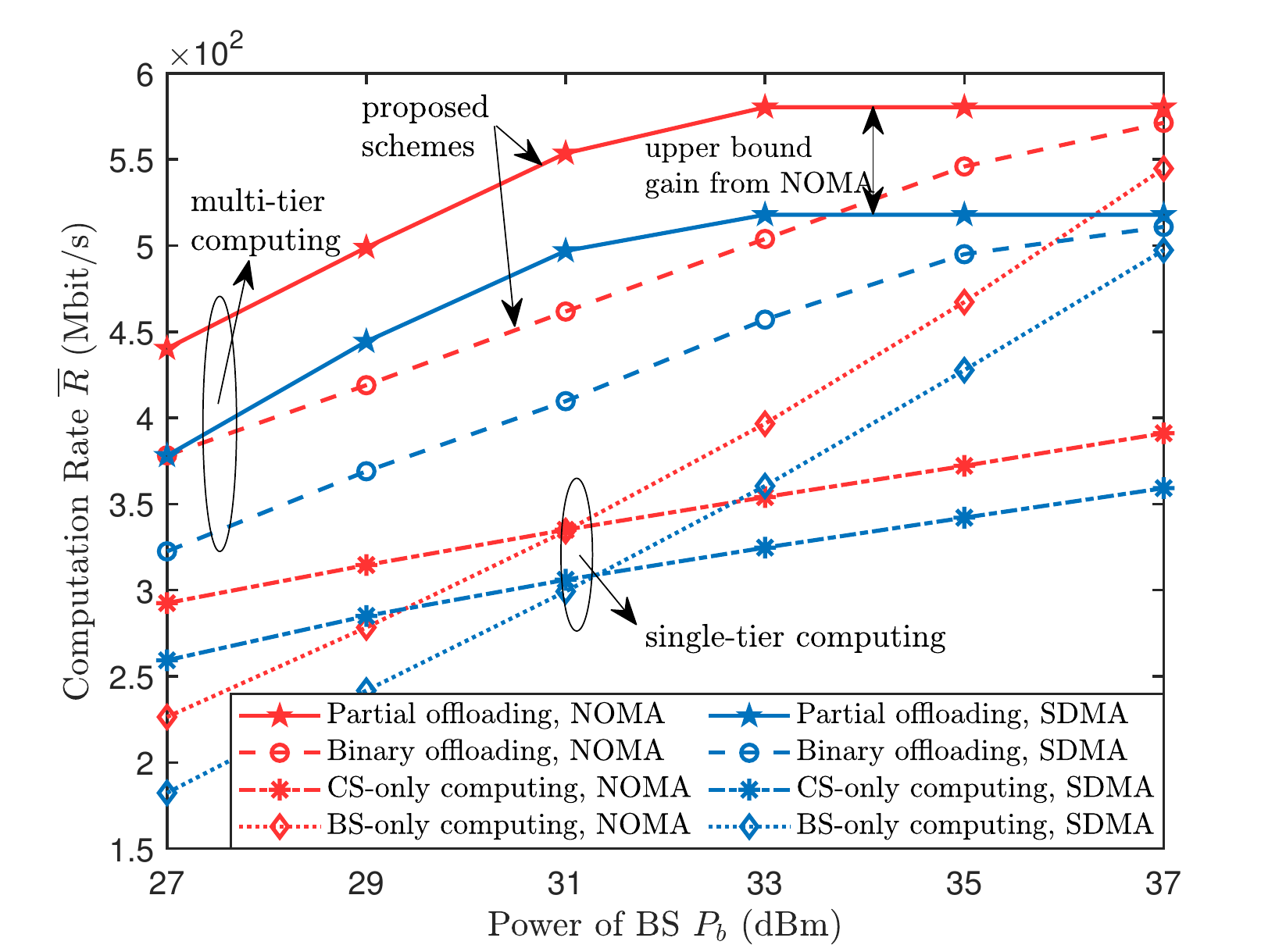}
    \caption{Computation rate $\tilde{R}$ versus power budget $P_b$, where $\gamma_{s,\min}=30$ dB.}
    \label{fig:rate_vs_power}
    \vspace{-0.4cm}
\end{figure}

In Fig. \ref{fig:rate_vs_SINR}, we present the computation rate achieved by different schemes versus the sensing SINR. As can be seen, the partial offloading mode achieves the best performance since it fully exploits the computation resources at both BS and CS. Furthermore, both partial offloading and binary offloading modes significantly outperform the CS-only and BS-only computing schemes, revealing the advantages of multi-tier computing when the capability of the BS is limited. It can also be observed that the computation rate obtained by all schemes decreases as $\gamma_{s}$ increases, which is consistent with the computing-sensing trade-off analyzed in \textbf{Remark \ref{remark_1}}. Finally, the employment of SIC allows NOMA-aided schemes to achieve a much higher computation rate than SDMA-aided schemes, especially when the sensing SINR is high. The above results verify the effectiveness of NOMA in the JCSMC framework, which is consistent with \textbf{Remark \ref{remark_2}}.

In Fig. \ref{fig:rate_vs_power}, we investigate the computation rate achieved by different schemes versus the total power available at the BS, where we set $\gamma_{s,\min} = 30$ dB. It is shown that, with the increase of $P_b$, the higher computation rate can be achieved by all the schemes because there are more resources available for both computing and offloading at the BS. However, the computation rate attained upon increasing $P_b$ is bounded by a certain value since it is limited by the communication rate $R_{u,k}, \forall k \in \mathcal{K}$ from the users to the BS. Furthermore, since NOMA is the capacity-achieving method for uplink communication, the NOMA-aided schemes attain a higher upper bound of the computation rate than the SDMA-aided schemes, which matches the analysis in \textbf{Remark \ref{remark_2}}. It is also interesting to observe that multi-tier computing leverages less power to reach the associated upper bounds, while the computation rate gain from multi-tier computing over single-tier computing, especially the BS-only computing, gradually becomes negligible when $P_b$ is sufficiently large. 

\begin{figure}[t!]
    \centering
    \includegraphics[width=0.4\textwidth]{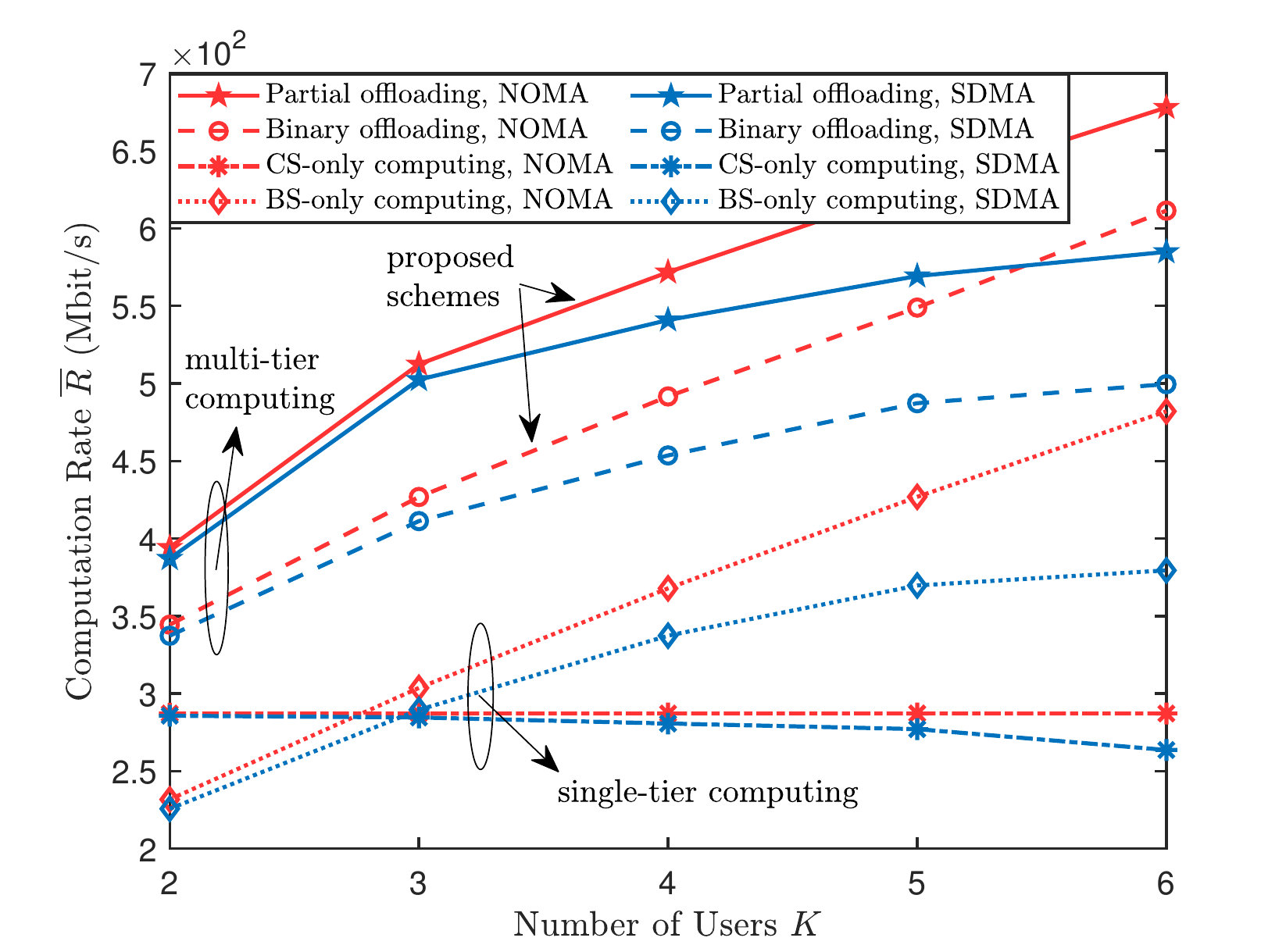}
    \caption{Computation rate $\tilde{R}$ versus the number of users $K$, where $\gamma_{s,\min}=30$ dB.}
    \label{fig:rate_vs_user_number}
    \vspace{-0.4cm}
\end{figure}
\begin{figure}[t!]
    \centering
    \includegraphics[width=0.4\textwidth]{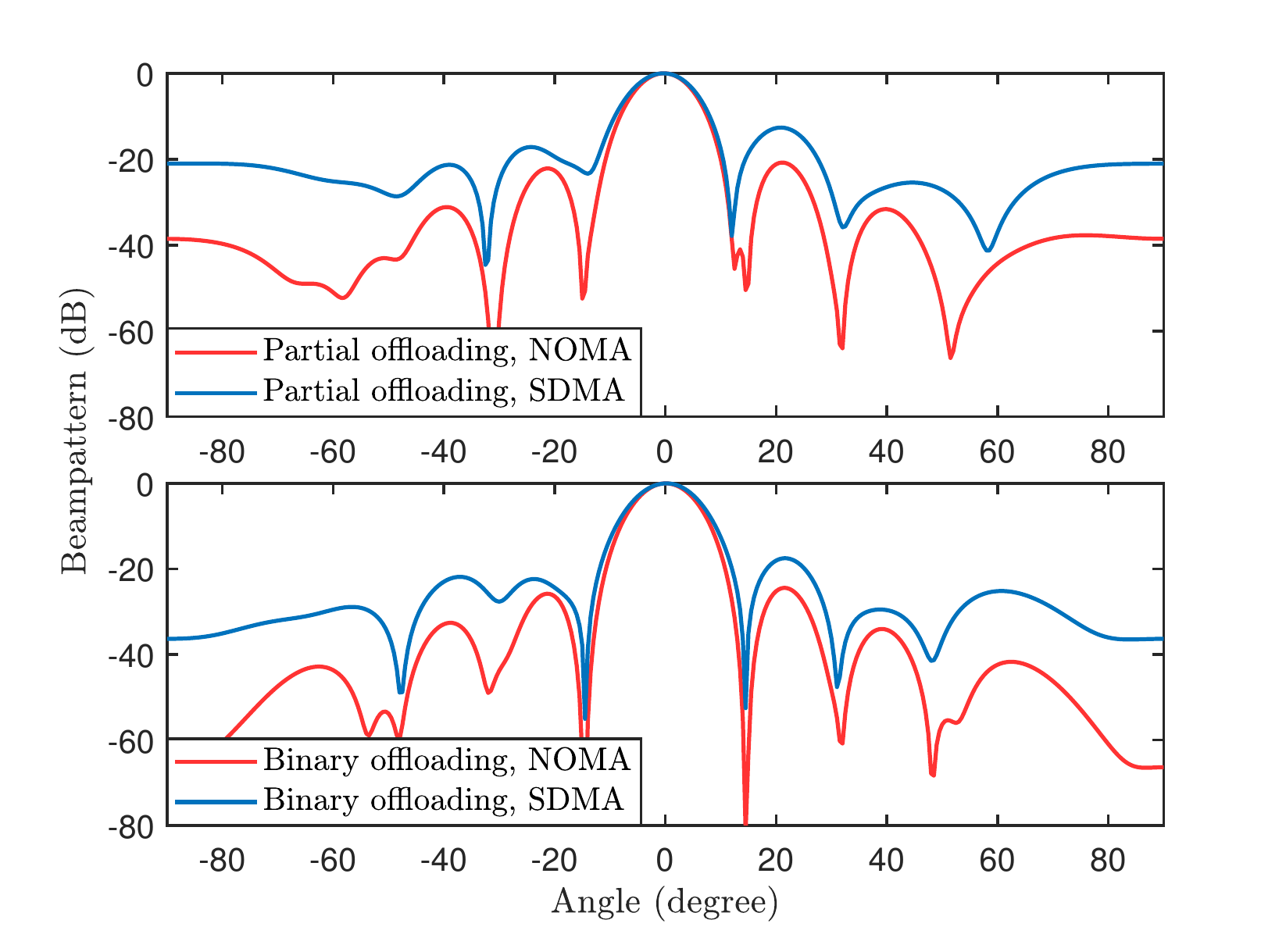}
    \caption{Normalized transmit-receive beampattern obtained by different schemes, where $\tilde{R} \approx 520$ Mbit/s for the partial offloading and $\tilde{R} \approx 420$ Mbit/s for the binary offloading.}
    \label{fig:beampattern}
    \vspace{-0.4cm}
\end{figure}

In Fig. \ref{fig:rate_vs_user_number}, we further demonstrate the benefits of the proposed framework compared with the benchmarks considering the increase of the number of users $K$, where $\gamma_{s,\min}$ is set as $30$ dB and the weights for all the users to the same value, i.e., $\omega_k = 1.0, \forall k \in \mathcal{K}$. We can observe that for the NOMA-aided schemes, the computation rate of partial offloading, binary offloading, and BS-only computing becomes higher as $K$ increases, while that of CS-only computing remains unchanged. The reason is that the performance of the CS-only computing is only determined by the BS$\rightarrow$CS communication rate when the user$\rightarrow$BS communication rate is sufficiently large. 
Conversely, the SDMA-aided schemes experience a computation rate degradation especially when $K$ is high. This is indeed expected because the sensing functionality required more power subject to stronger interference from more computation users. The above results also confirm the necessity of NOMA in scenarios with massive connectivity.

In Fig. \ref{fig:beampattern}, we plot the normalized transmit-receive beampattern of sensing over an angle grid $\mathbf{\Phi} = [-\frac{\pi}{2}:\frac{\pi}{100}:\frac{\pi}{2}]$ via different schemes. The results are obtained over one random channel realization. The normalized transmit-receive beampattern is defined as
\begin{equation}
    P(\theta_i) = \frac{|\mathbf{w}_s^\star \mathbf{a}(\theta_i) \mathbf{a}^T(\theta_i) \mathbf{p} |^2}{ \max_{\theta_j \in \mathbf{\Phi}} |\mathbf{w}_s^\star \mathbf{a}(\theta_i) \mathbf{a}^T(\theta_i) \mathbf{p}|^2 }, \forall \theta_i \in \mathbf{\Phi},
\end{equation}
which indicates the array gain at the receiver with the transmit beamformer $\mathbf{p}$ in the direction of $\theta_i$. Fig. \ref{fig:beampattern} shows that when the computation rate is approximately the same, both NOMA-aided and SDMA-aided schemes can achieve the mainlobe in the target direction. It can also be seen that the interference from the clutter directions can be well suppressed in the NOMA-aided frameworks, leading to robust sensing performance, while the SDMA-aided framework has a lower peak-to-sidelobe ratio, thus suffering stronger interference and resulting in lower sensing performance. 

\section{Conclusions} \label{sec:conclusion}

In this paper, a NOMA-aided joint communication, sensing, and multi-tier computing framework was proposed. Two computation rate maximization problems subject the sensing SINR constraints were formulated considering both partial offloading and binary offloading modes. A WMMSE-based AO algorithm and an ADMM-based AO algorithm were proposed to address the two optimization problems, respectively. Numerical results revealed that multi-tier computing significantly outperforms single-tier computing especially when the resources at the BS are limited. It was also shown that NOMA plays an important role in guaranteeing high-quality computing and sensing performance simultaneously compared with the conventional transmission scheme.     

\section*{Appendix~A: Proof of Proposition \ref{proposition_1}} \label{appendix_2}

The Lagrangian of the problem (P1.2) can be expressed as 
\begin{align}
    &\mathcal{L} = \nonumber \\
    &-\sum_{k \in \mathcal{K}} \omega_k \left( r_{b,k} + r_{c,k} \right)  + \sum_{k \in \mathcal{K}} \mu_k \underbrace{(r_{b,k}+r_{c,k}-B c + B \varepsilon_{u,k})}_{\mathcal{G}_{u,k}} \nonumber\\
    &+ \mu_d \underbrace{\left(\sum_{k \in \mathcal{K}} r_{c,k} - B c + B \varepsilon_d \right)}_{\mathcal{G}_d}
    + \mu_s \underbrace{(\Gamma_{s,\min} - c + \varepsilon_s)}_{\mathcal{G}_s} + \mathcal{G}_0,
\end{align}
where $\boldsymbol{\mu} = [\mu_1,\dots,\mu_K, \mu_d, \mu_s]^T$ denote the Lagrangian multipliers associated with the constraints \eqref{constraint:up_rate_mmse}-\eqref{constraint:sensing_mmse} and $\mathcal{G}_0$ denotes all the terms corresponding to the rest of the constraints. Similarly, the Lagrangian of the problem (P1.1) is given by 
\begin{align}
    \tilde{\mathcal{L}}
    & = -\sum_{k \in \mathcal{K}} \mu_k \left( r_{b,k} + r_{c,k} \right)  + \sum_{k \in \mathcal{K}} \mu_k \underbrace{(r_{b,k}+r_{c,k}-B R_{u,k})}_{\tilde{\mathcal{G}}_{u,k}} \nonumber\\
    &+ \mu_d \underbrace{\left(\sum_{k \in \mathcal{K}} r_{c,k} - B R_d \right)}_{\tilde{\mathcal{G}}_d}
    + \mu_s \underbrace{(\Gamma_{s,\min} - R_s)}_{\tilde{\mathcal{G}}_s} + \mathcal{G}_0.
\end{align}
Suppose $\{\boldsymbol{\varpi}^\star, \tilde{\mathbf{w}}^\star, \mathbf{r}_b^\star, \mathbf{r}_c^\star, \mathbf{p}^\star\}$ is a KKT optimal solution to problem (P1.2) and $\boldsymbol{\mu}^\star$ are the corresponding optimal Lagrangian multipliers. Then, \textbf{Proposition \ref{proposition_1}} claims that if $\{\boldsymbol{\varpi}^\star, \tilde{\mathbf{w}}^\star\}$ satisfies \eqref{condition:weight} and \eqref{condition:receiver}, $\{\mathbf{r}_b^\star, \mathbf{r}_c^\star, \mathbf{p}^\star\}$ is a KKT optimal solution to problem (P1.1). Firstly, it is easy to verify that with the optimal weights and receivers given in \eqref{condition:weight} and \eqref{condition:receiver}, the stationarity conditions associated with the variables $\boldsymbol{\varpi}$ and $\tilde{\mathbf{w}}$ of problem (P1.2) are satisfied. Thus, we only focus on the remaining KKT conditions in the following.

Firsty, we show that the stationarity conditions associated with the variables $\{\mathbf{r}_b, \mathbf{r}_c, \mathbf{p}\}$ of problem (P1.2) imply the corresponding stationarity conditions of problem (P1.1). The stationarity conditions of problem (P1.2) are given by 
\begin{align}
    &\frac{\partial \mathcal{L}}{\partial r_{b,k}} = -\mu_k + \mu_k + \frac{\partial \mathcal{G}_0}{\partial r_{b,k}}=0, \forall k \in \mathcal{K},\\
    &\frac{\partial \mathcal{L}}{\partial r_{c,k}} = -\mu_k + \mu_k + \mu_d + \frac{\partial \mathcal{G}_0}{\partial r_{b,k}}=0, \forall k \in \mathcal{K},\\
    \label{kkt_p}
    &\nabla_{\mathbf{p}} \mathcal{L} = \mu_k B \nabla_{\mathbf{p}} \varepsilon_{u,k} + \mu_d B \nabla_{\mathbf{p}} \varepsilon_d + \mu_s \nabla_{\mathbf{p}} \varepsilon_s + \nabla_{\mathbf{p}} \mathcal{G}_0 = \mathbf{0}.
\end{align}
It is easy to verify that $\frac{\partial \tilde{\mathcal{L}}}{\partial r_{b,k}} = \frac{\partial \mathcal{L}}{\partial r_{b,k}}=0$ and $\frac{\partial \tilde{\mathcal{L}}}{\partial r_{c,k}} = \frac{\partial \mathcal{L}}{\partial r_{c,k}}=0$. Therefore, the stationarity conditions associated with $\mathbf{r}_b$ and $\mathbf{r}_c$ of problem (P1.1) are satisfied. Furthermore, according to the rate-WMMSE relationships in \eqref{eqn:comm_rate_wmmse}, we have $R_{u,k} = \max_{\varpi_{u,k}, \tilde{\mathbf{w}}_{u,k}} \left( c - \varepsilon_{u,k} \right)$, the unique optimal solution of which satisfies the conditions in \eqref{condition:weight} and \eqref{condition:receiver}, i.e., $\varpi_{u,k}^\star = \varpi_{u,k}^{\text{MMSE}}(\mathbf{p})$ and $\tilde{\mathbf{w}}_k^\star = \tilde{\mathbf{w}}_k^{\text{MMSE}}(\mathbf{p})$. Then, by applying the Danskin's theorem \cite{Bertsekas1999nonlinear}, it holds that $\nabla_{\mathbf{p}} R_{u,k} = -\nabla_{\mathbf{p}} \varepsilon_{u,k}(\varpi_{u,k}^\star, \tilde{\mathbf{w}}_{u,k}^\star)$. 
Following the same procedure, it can also be proved that $\nabla_{\mathbf{p}} R_d = -\nabla_{\mathbf{p}} \varepsilon_d(\varpi_d^\star, \tilde{\mathbf{w}}_d^\star)$ and $\nabla_{\mathbf{p}} R_s = -\nabla_{\mathbf{p}} \varepsilon_d(\varpi_s^\star, \tilde{\mathbf{w}}_s^\star)$. Thus, given that $ \{\boldsymbol{\varpi}^\star, \tilde{\mathbf{w}}^\star, \mathbf{r}_b^\star, \mathbf{r}_c^\star, \mathbf{p}^\star\}$ satisfies \eqref{condition:weight}, \eqref{condition:receiver}, and \eqref{kkt_p}, $ \{\mathbf{r}_b^\star, \mathbf{r}_c^\star, \mathbf{p}^\star\}$ satisfies the following condition:
\begin{align}
    \nabla_{\mathbf{p}} \tilde{\mathcal{L}} = &-\mu_k B \nabla_{\mathbf{p}} R_{u,k} - \mu_d B \nabla_{\mathbf{p}} R_d - \mu_s \nabla_{\mathbf{p}} R_s + \nabla_{\mathbf{p}} \mathcal{G}_0 \nonumber \\
    = &\nabla_{\mathbf{p}} \mathcal{L} = \mathbf{0},
\end{align}
which indicates that the stationarity condition associated with $\mathbf{p}$ of problem (P1.1) is satisfied.

Next, we show that the feasibility and complementarity conditions of problem (P1.2) imply the corresponding conditions of problem (P1.1), where we only focus on the different constraints between (P1.2) and (P1.1), i.e., \eqref{constraint:up_rate_mmse}-\eqref{constraint:sensing_mmse} and \eqref{constraint:1.1.1}-\eqref{constraint:1.1.3}. Take \eqref{constraint:up_rate_mmse} and \eqref{constraint:1.1.1} as an example. The feasibility and complementarity conditions associated with \eqref{constraint:up_rate_mmse} of (P1.2) are given by 
\begin{gather}
    \label{kkt_complementarity}
    \mu_k \ge 0, \quad \mathcal{G}_{u,k} \le 0, \quad \mu_k \mathcal{G}_{u,k}  = 0.
\end{gather}  
By applying the conditions in \eqref{condition:weight} and \eqref{condition:receiver} and the rate-WMMSE relationship in \eqref{eqn:comm_rate_wmmse}, it can be verified that $\tilde{\mathcal{G}}_{u,k} = \mathcal{G}_{u,k}(\varpi_{u,k}^\star, \tilde{\mathbf{w}}_{u,k}^\star)$. Thus, given that $ \{\boldsymbol{\varpi}^\star, \tilde{\mathbf{w}}^\star, \mathbf{r}_b^\star, \mathbf{r}_c^\star, \mathbf{p}^\star\}$ satisfies \eqref{condition:weight}, \eqref{condition:receiver}, and \eqref{kkt_complementarity}, it holds for $ \{\mathbf{r}_b^\star, \mathbf{r}_c^\star, \mathbf{p}^\star\}$ that 
\begin{gather}
    \mu_k \ge 0, \quad \tilde{\mathcal{G}}_{u,k} \le 0, \quad \mu_k \tilde{\mathcal{G}}_{u,k} = 0,
\end{gather}
which indicates that the feasibility and complementarity conditions associated with \eqref{constraint:1.1.1} of (P1.1) are satisfied. Following the same procedure, it can also be proved that the conditions associated with \eqref{constraint:1.1.2} and \eqref{constraint:1.1.3} of (P1.1) are also satisfied, which completes the proof.

\section*{Appendix~B: Proof of Proposition \ref{proposition_2}} \label{appendix_2}

Firstly, we demonstrate the convergence of \textbf{Algorithm \ref{alg:A}}. It can be observed that solving problem (P1.2) via \textbf{Algorithm \ref{alg:A}} is equivalent to solving the following problem to obtain $\{\mathbf{r}_b^{[n+1]}, \mathbf{r}_c^{[n+1]},$ $\mathbf{p}^{[n+1]}\}$ at the $n$-th iteration: 
\begin{subequations}
\begin{align}
    &\text{(P1.2.n):} \quad \max_{\mathbf{r}_b, \mathbf{r}_c, \mathbf{p}} \quad  \sum_{k \in \mathcal{K}} \omega_k \left(r_{b,k} + r_{c,k} \right) \\
    \label{constraint:iteration_n}
    \mathrm{s.t.} \quad & r_{b,k} + r_{c,k} \le B c - B \varepsilon_{u,k}(\varpi_{u,k}^{[n+1]}, \tilde{\mathbf{w}}_{u,k}^{[n+1]}, \mathbf{p}), \forall k \in \mathcal{K}, \\
    & \sum_{k \in \mathcal{K}} r_{c,k}\le B c - B \varepsilon_d(\varpi_d^{[n+1]}, \tilde{\mathbf{w}}_d^{[n+1]}, \mathbf{p}) \\
    & \Gamma_{s,\min} \le c - \varepsilon_s(\varpi_s^{[n+1]}, \tilde{\mathbf{w}}_s^{[n+1]}, \mathbf{p}), \\
    & \eqref{constraint:resource_compute}-\eqref{constraint:positive},
\end{align}
\end{subequations}
Take constraint \eqref{constraint:iteration_n} as an example. According to the rate-WMMSE relationships, it holds that 
\begin{align}
&c - \varepsilon(\varpi_{u,k}^{[n+1]}, \tilde{\mathbf{w}}_{u,k}^{[n+1]}, \mathbf{p}) \le R_{u,k}(\mathbf{p}), \\
&c - \varepsilon(\varpi_{u,k}^{[n+1]}, \tilde{\mathbf{w}}_{u,k}^{[n+1]}, \mathbf{p}^{[n]}) = R_{u,k}(\mathbf{p}^{[n]}).
\end{align}
Then, it follows that 
\begin{align}
&r_{b,k}^{[n+1]} + r_{c,k}^{[n+1]} \overset{(a)}{\le} B c- B \varepsilon(\varpi_{u,k}^{[n+1]}, \tilde{\mathbf{w}}_{u,k}^{[n+1]}, \mathbf{p}^{[n+1]}) \nonumber\\
&\overset{(b)}{\le}  B R_{u,k}(\mathbf{p}^{[n+1]}) \overset{(c)}{=} B c - B \varepsilon_k (\varpi_{u,k}^{[n+2]}, \tilde{\mathbf{w}}_{u,k}^{[n+2]}, \mathbf{p}^{[n+1]}),
\end{align}
where $(a)$ is due to the feasibility of $\{\mathbf{r}_b^{[n+1]}, \mathbf{r}_c^{[n+1]}, \mathbf{p}^{[n+1]}\}$ to problem (P1.2.n), and $(b)$ and $(c)$ stem from the rate-WMMSE relationships. Thus, we have proved that the solution $\{\mathbf{r}_b^{[n+1]}, \mathbf{r}_c^{[n+1]}, \mathbf{p}^{[n+1]}\}$ at the $n$-th iteration also satisfies the constraint \eqref{constraint:iteration_n} of problem at the $(n+1)$-th iteration. Following the same path and checking all the constraints, it is easy to prove that the solution $\{\mathbf{r}_b^{[n+1]}, \mathbf{r}_c^{[n+1]}, \mathbf{p}^{[n+1]}\}$ is also a feasible solution to the problem at $(n+1)$-th iteration, which guarantees that the objective value is non-decreasing after each iteration in Algorithm \ref{alg:A}. Since the objective function is bounded due to the power constraint \eqref{constraint:resource_compute}, the convergence of \textbf{Algorithm \ref{alg:A}} is guaranteed.

Consequently, the sequence $\{\boldsymbol{\varpi}^{[n]}, \tilde{\mathbf{w}}^{[n]}, \mathbf{r}_b^{[n]}, \mathbf{r}_c^{[n]}, \mathbf{p}^{[n]}\}_{n=0}^\infty$ must converge to a cluster point $\{\boldsymbol{\varpi}^\star, \tilde{\mathbf{w}}^\star,$ $ \mathbf{r}_b^\star, \mathbf{r}_c^\star, \mathbf{p}^\star\}$. Since the map $\tilde{\mathbf{w}}_{u,k}^{\text{MMSE}}(\cdot)$ are continuous, it holds that
\begin{equation}
    \tilde{\mathbf{w}}_{u,k}^\star = \lim_{n \rightarrow \infty} \tilde{\mathbf{w}}_{u,k}^{\text{MMSE}}(\mathbf{p}^{[n]}) = \tilde{\mathbf{w}}_{u,k}^{\text{MMSE}}(\mathbf{p}^\star).
\end{equation} 
Based on the same idea, it can be proved that the cluster point satisfies \eqref{condition:weight} and \eqref{condition:receiver}. Thus, according to \textbf{Proposition \ref{proposition_1}}, $\boldsymbol{\varpi}^\star$ and $\tilde{\mathbf{w}}^\star$ satisfy the KKT conditions of (P1.2). Given the optimal $\boldsymbol{\varpi}^\star$ and $\tilde{\mathbf{w}}^\star$, the $\{\mathbf{r}_b^\star, \mathbf{r}_c^\star, \mathbf{p}^\star\}$ is a globally optimal solution to the following convex optimization problem:
\begin{subequations}
\begin{align}
    \text{(P1.3):} &\quad \max_{\mathbf{r}_b, \mathbf{r}_c, \mathbf{p}} \quad \sum_{k \in \mathcal{K}} \omega_k \left(r_{b,k} + r_{c,k} \right) \\
    \mathrm{s.t.} \quad & r_{b,k} + r_{c,k} \le B c - B \varepsilon_{u,k}(\varpi_{u,k}^\star, \tilde{\mathbf{w}}_{u,k}^\star, \mathbf{p}), \forall k \in \mathcal{K}, \\
    & \sum_{k \in \mathcal{K}} r_{c,k}\le B c - B \varepsilon_d(\varpi_d^\star, \tilde{\mathbf{w}}_d^\star, \mathbf{p}) \\
    & \Gamma_{s,\min} \le c - \varepsilon_s(\varpi_s^\star, \tilde{\mathbf{w}}_s^\star, \mathbf{p}), \\
    & \eqref{constraint:resource_compute}-\eqref{constraint:positive},
\end{align}
\end{subequations}    
By checking the KKT conditions of problem (P1.3), it can be proved that $\{\mathbf{r}_b^\star, \mathbf{r}_c^\star, \mathbf{p}^\star\}$ also satisfies the KKT conditions of problem (P1.2), which completes the proof.

\bibliographystyle{IEEEtran}
\bibliography{reference/mybib}

\end{document}